\documentclass[aps,prd,twocolumn,preprintnumbers,nofootinbib]{revtex4-2}
\usepackage{amsmath, mathrsfs, amssymb,amsfonts,amsthm,graphicx, epsf, dcolumn, yfonts, subfigure,thmtools}
\usepackage[hyperfootnotes=true]{hyperref}
\usepackage{color}
\usepackage{slashed}
\usepackage{setspace}
\usepackage{cancel}
\usepackage{wasysym}
\usepackage{comment}

\usepackage{float}
\pdfoutput=1
\newcommand{\be}{\begin{equation}}
\newcommand{\PP}{\mathcal{P}}
\newcommand{\ee}{\end{equation}}
\newcommand{\bea}{\begin{eqnarray}}
\newcommand{\eea}{\end{eqnarray}}

\newcommand{\Vol}{\textup{Vol}}

\newcommand{\beq}{\begin{equation}}
\newcommand{\eeq}{\end{equation}}

\renewcommand*{\thefootnote}{\fnsymbol{footnote}}

\begin{document}

\preprint{\texttt{WI-42-2025, IFT-UAM/CSIC-25-156}}

\title{Lorentzian threads and nonlocal computation in holography}
\author{Elena C\'aceres,$^{1}$ Rafael Carrasco$^{2}$ and Juan F. Pedraza$^{2}$}
\affiliation{$^1$Theory Group, Department of Physics, The University of Texas at Austin, Austin, Texas 78712, USA\\
 $^2$Instituto de F\'isica Te\'orica UAM/CSIC,
Calle Nicol\'as Cabrera 13-15, Cantoblanco, 28049 Madrid, Spain}

\begin{abstract}\vspace{-2mm}
Recent advances in holography and quantum gravity have shown that CFTs with classical gravity duals can implement nonlocal quantum computation protocols that appear local from the bulk perspective. We examine the extent to which current prescriptions for holographic complexity support this claim, focusing on the Complexity=Volume (CV) proposal. The reformulation of CV in terms of Lorentzian threads suggests that bulk computations are performed with local gates. However, we find that the original formalism is insufficient when it comes to analyzing the complexity of subsystems and their inequalities. Specifically, standard Lorentzian threads cannot account for the negativity of `mutual complexity' and its higher-partite generalizations. To address this deficiency, we modify the Lorentzian threads program by introducing multiple flavors of threads. Our analysis reveals that an optimal solution for this new program implies the existence of additional types of gates that enable nonlocal computations in the dual CFT. We give a tentative interpretation of the multiflavor program in terms of Lorentzian `hyperthreads,' in analogy with the Riemannian case.
\end{abstract}

\renewcommand*{\thefootnote}{\arabic{footnote}}
\setcounter{footnote}{0}

\maketitle

\section{Introduction}

A major conceptual shift in our understanding of holography has come from viewing gravitational dynamics through the lens of quantum information. The Ryu–Takayanagi (RT) formula \cite{Ryu:2006bv,Ryu:2006ef}, which relates the entanglement entropy in a CFT to the area of an extremal codimension-2 bulk surface, has revealed deep connections between geometry and entanglement. This insight has driven significant progress on the black hole information paradox \cite{Penington:2019npb,Chen:2020hmv,Ghosh:2021axl,Almheiri:2019hni}, bulk reconstruction \cite{Czech:2012bh,Dong:2016eik,Czech:2016xec,Espindola:2018ozt,Bao:2019bib,Guijosa:2022jdo,Leutheusser:2022bgi} and the emergence of spacetime itself \cite{VanRaamsdonk:2010pw,Swingle:2009bg,Lashkari:2013koa,Faulkner:2013ica}.

Entanglement entropy alone, however, does not capture all aspects of holographic spacetimes. In particular, entanglement fails to encode the dynamical behavior of black hole interiors, notably the late-time growth of the Einstein–Rosen bridge \cite{Susskind:2014moa}. This limitation points to the need for a holographic observable capable of probing interior dynamics. A natural candidate is \emph{complexity}, a quantity that captures the computational structure of quantum states \cite{Susskind:2014rva}. Our understanding of complexity in quantum field theory, and of its holographic realization, remains far less complete than that of entanglement, and numerous questions are still open. In this work we investigate the subsystem complexity of holographic systems through the framework of Lorentzian threads \cite{Pedraza:2021mkh,Pedraza:2021fgp}.

A common notion of complexity in quantum mechanical systems is \emph{circuit complexity}, defined as the minimal number of elementary operations (or gates) in a quantum circuit that takes a chosen reference state into a given target state. In quantum field theories, however, a precise and universally accepted definition of complexity is still elusive \cite{Chapman:2017rqy,Jefferson:2017sdb}, owing to various scheme-dependent ambiguities in its formulation. On the gravitational side, several holographic proposals have been put forward to capture the salient features of complexity \cite{Stanford:2014jda,Brown:2015bva,Couch:2016exn,Belin:2021bga}. Among them, the Complexity=Volume (CV) proposal \cite{Stanford:2014jda} bears a close resemblance to the RT prescription for entanglement entropy, providing a complexity-based counterpart to entanglement-based notions of spacetime emergence \cite{Pedraza:2022dqi,Carrasco:2023fcj}. According to CV, the complexity of a CFT state defined on a Cauchy slice $\sigma=\partial A$ (with $A\subset\partial \mathcal{M}$ and $\mathcal{M}$ the bulk manifold) is given by the volume of the maximal bulk codimension-1 surface homologous to $A$,
\begin{equation}
\mathcal{C}(\sigma) = \frac{1}{G_N \ell}
\max_{\Sigma\sim A} \text{Vol}(\Sigma)\,, 
\label{eq:CVduality}
\end{equation}
where $G_N$ is Newton's constant and $\ell$ a length scale.

Beyond their geometric formulations, both entanglement entropy and the CV proposal admit alternative descriptions as convex optimization problems \cite{Headrick:2017ucz}. For instance, the RT formula \cite{Ryu:2006bv,Ryu:2006ef}, which gives the entanglement entropy of a boundary region $A$ on a (time-reflection-symmetric) Cauchy slice as the minimal area of a codimension-2 surface homologous to $A$,
\begin{equation}
S(A) = \frac{1}{4G_N}
\min_{\Gamma\sim A} \text{Area}(\Gamma)\,, 
\label{eq:RT}
\end{equation}
can be cast as the maximal flux of a divergenceless vector field $v$ through $A$, subject to the norm bound $|v|\leq 1$ \cite{Freedman:2016zud}:
\begin{equation}
\hspace{-5pt}S(A) = \frac{1}{4G_N}\max_{v\in\mathcal{F}}\!\int_A\!\! v,
\,\,\; \mathcal{F}
\equiv\left\{v\,\big|\, \nabla\cdot v=0\,,\, |v|\leq 1\right\}.\hspace{-2pt}
\label{eq:RTthreads}
\end{equation}
For any optimal solution $v$, the integral curves of this
vector field define the so-called \emph{bit threads} \cite{Freedman:2016zud}: Planck-thickness flow lines
connecting the region $A$ to its complement $A^c$. Intuitively, one may think of each thread as representing a
pair of maximally entangled qubits shared between $A$
and $A^c$. In this language, the RT prescription becomes a
maximization problem that counts the maximal number
of Bell pairs that may be distilled between $A$ and $A^c$ while
respecting the local density bound \cite{Headrick:2022nbe}. Thus, the bit-thread reformulation provides a more microscopic account of holographic entanglement while remaining formally equivalent to the RT prescription.

An important feature of entanglement entropy is the existence of the so-called entropy inequalities, which constrain the entropies of multipartite systems. The simplest ones are subadditivity and strong subadditivity, given by
\begin{equation}
S(A_1) + S(A_2) \geq S(A_1\cup A_2)\,,
\end{equation}
\vspace{-7mm}
\begin{equation}
S(A_1\cup A_2) + S(A_2\cup A_3)\geq  S(A_2) + S(A_1\cup A_2 \cup A_3)\,,
\end{equation}
and are valid for arbitrary quantum systems. In holography, the RT prescription implies an infinite set of additional inequalities, which can be used to characterize and discriminate theories with possible gravity duals. A notable example is the monogamy of mutual information \cite{Hayden:2011ag}, which for a tripartite system imposes the constraint
\begin{equation}
    \begin{split}
        &\,\,S(A_1\cup A_2) + S(A_1\cup A_3) + S(A_2\cup A_3) \geq \\
       &S(A_1) + S(A_2) + S(A_3)+S(A_1\cup A_2\cup A_3)\,.
    \end{split}
\end{equation}
More generally, for a system partitioned into $N$ components, the complete set of entropy inequalities defines an $M$-dimensional space, where $M = \sum_{k}\binom{N}{k}$ counts the possible nonempty unions of subsystems. This region of allowed entropy configurations is called the \emph{holographic entropy cone} \cite{Bao:2015bfa}, and its intricate structure has been the subject of an extensive body of work in the literature, with far too many contributions to cite individually here.

While these inequalities have a clear geometric counterpart from the RT perspective, proving them in the language of bit threads appears to encounter serious obstructions and may well be impossible in general. For instance, the bit-thread picture is insufficient to simultaneously reproduce the entropies of all individual regions and of all their arbitrary unions, which would be a necessary step in such a proof. A notable exception is the monogamy of mutual information \cite{Cui:2018dyq,Hubeny:2018bri,Agon:2018lwq}, whose validity is guaranteed by the multicommodity theorem for Riemannian flows \cite{Cui:2018dyq}. This theorem states that it is possible to construct a flow that simultaneously reproduces the entropies of all individual regions and the full system, which, in the tripartite case, provides all the required ingredients. An alternative route is to generalize the bit-thread prescription to allow for hyperthreads (or $k$-threads) \cite{Harper:2021uuq,Harper:2022sky}. In this case, flow lines are allowed to split into branches that terminate on $k$ distinct boundary subregions, enabling the encoding of genuine multipartite correlations. From this perspective, a $k$-thread can be interpreted as a fundamental unit of $k$-party entanglement, which is particularly useful for describing situations where the mostly-bipartite conjecture fails \cite{Akers:2019gcv,Balasubramanian:2024ysu,Iizuka:2025ioc,Iizuka:2025bcc,Iizuka:2025caq}.

Given the central role of entropy inequalities and the entropy cone in constraining holographic entanglement, it is natural to ask whether there exist analogous inequalities for subregion complexities that define a \emph{complexity cone}. Any such inequalities are expected to be non-universal, with each set characterizing a distinct bulk prescription for holographic complexity; see \cite{Agon:2018zso,Caceres:2018blh,Caceres:2019pgf} for some preliminary ideas. Despite this expected non-universality, they could still be extremely useful, providing criteria to map the broad landscape of complexity measures from both the bulk and boundary perspectives \cite{Belin:2021bga} (see also \cite{Belin:2022xmt,Myers:2024vve,Caceres:2025myu}). Closely related is the question of whether such inequalities encode a statement about complexity analogous to the mostly-bipartite conjecture for entanglement, or whether no such analogue exists.

Motivated by these questions, we work within the CV proposal and use its reformulation as an optimization problem involving flows in the bulk. In this language, the relevant objects are collections of \emph{timelike} curves, known as \emph{Lorentzian threads}. These curves play a role analogous to that of bit threads for entanglement, but are now interpreted as \emph{gatelines} representing the elementary operations (or gates) needed to construct a target state from a reference state \cite{Pedraza:2021mkh,Pedraza:2021fgp}. When applied to the complexity of the full boundary state (in analogy with the entanglement of a bipartite system), the Lorentzian-thread reformulation naturally suggests that holographic states can be prepared using mostly local gates; this may be regarded as the complexity analogue of the mostly-bipartite conjecture. However, there is robust evidence that this expectation should not hold in general \cite{May:2019odp,May:2022rko,May:2022clu,Dolev:2022gwj,May:2023kfp}. Notably, the connected wedge theorem implies the existence of intrinsically nonlocal quantum computation protocols implemented through bulk dynamics.

To probe this tension more systematically, we first ask whether Lorentzian threads can be used to define and compute subsystem complexity for a partition of the boundary Cauchy slice $\sigma$. In Section~\ref{sec: lorentzian threads}, we review the Lorentzian-thread formulation of CV and identify important shortcomings: in particular, we show that even for bipartite configurations the naive prescription fails to reproduce the expected \emph{superadditivity} of subregion complexities. We trace this failure to the absence of a multicommodity theorem in Lorentzian settings. In Section~\ref{sec: multiple flavors}, we then propose a new program for an arbitrary number of subregions based on families of Lorentzian threads with distinct \emph{flavors}, each associated with a specific subregion, and we show that this multi-flavor framework can simultaneously compute the complexity of individual subregions and of the full state. In Section~\ref{interpretation}, we analyze a discretized version of this program in order to clarify its microscopic interpretation. While the threads of the multi-flavor construction still appear to represent local gates, these basic operations are found to violate superadditivity at the microscopic level. We resolve this issue by a simple change of basis: suitable linear combinations of the basic operations can be made consistent with micro-superadditivity while remaining formally equivalent to CV. These combinations are interpreted as \emph{generalized Lorentzian hyperthreads}, a new type of elementary operation that enables nonlocal computations in the dual CFT. We conclude in Section~\ref{discussion} with a summary of our findings and a discussion of future directions.

\section{Lorentzian threads: a primer}\label{sec: lorentzian threads}

In this section, we review the reformulation of the Complexity=Volume (CV) prescription using flows, as proposed in \cite{Pedraza:2021mkh,Pedraza:2021fgp}. This approach stems from the continuous version of the Lorentzian min flow-max cut theorem introduced in \cite{Headrick:2017ucz}. According to this framework, the holographic complexity $\mathcal{C}(\sigma)$ of a state on a boundary Cauchy slice $\sigma$ is determined by the minimum flux of a divergenceless, norm-bounded vector field $v$ through a boundary surface $A$ anchored to $\sigma$. This formulation offers a more intuitive, microscopic interpretation of CV. Specifically, discretized flows correspond to Lorentzian threads ---timelike worldlines of microscopic thickness that represent unitary operations within a tensor network discretizing spacetime. These Lorentzian threads are essential for constructing bulk analogs of quantum circuits, facilitating the time evolution of dual CFT states \cite{Pedraza:2021mkh,Pedraza:2021fgp}.

After addressing certain subtleties of this reformulation, we will rewrite the flow prescription in terms of measure theory, following \cite{Headrick:2022nbe}. This measure-based approach avoids reliance on vector fields and proves particularly useful in later sections when extending the formalism to accommodate multiple thread flavors. We will conclude this section by discussing the application of this formalism and highlighting some limitations thereof.

\subsection{Flow prescription \& interpretation}\label{Min flow-max cut theorem in complexity}

Let us begin by establishing some notation and conventions. Let $\mathcal{M}$ be a $d$-dimensional, compact, oriented, time-oriented Lorentzian manifold with boundary $\partial \mathcal{M}$. We consider a boundary region $A\subset\partial \mathcal{M}$, such that its causal future coincides with itself, $J^+(A)\cap \partial \mathcal{M}=A$.
Such a region is delimited by a boundary Cauchy slice $\sigma(A)=\partial A$. We then consider the set of all bulk Cauchy slices $\Sigma(A)$ that are homologous to $A$ and, thus, anchored to $\sigma(A)$. Each of these surfaces is a compact, orientable codimension-1 hypersurface-with-boundary, that is everywhere either spacelike or lightlike. Together, these surfaces foliate a bulk causal diamond $D(A)$, anchored to $\sigma(A)$, also known as its  Wheeler–DeWitt (WdW) patch ---see Fig.~\ref{domain} for an illustration.

\begin{figure}
    \includegraphics{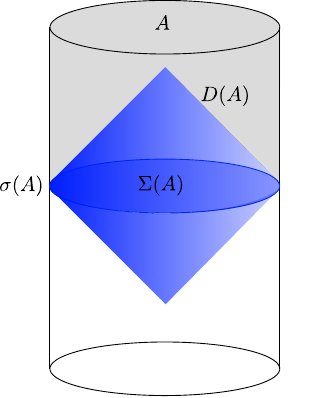}
    \caption{Boundary region $A$, boundary Cauchy slice $\sigma(A)$, and bulk Cauchy slice $\Sigma(A)$. The WdW patch associated with this Cauchy slice, depicted in blue, is a bulk causal diamond $D(A)$ anchored to $\sigma(A)$.}
    \label{domain}
\end{figure}

Next, we define a flow $v$. A flow is a timelike vector field on $\mathcal{M}$ that is (i) future directed, (ii) divergenceless and (iii) norm bounded from below. Mathematically, we state these conditions by defining the set $\mathcal{F}$ such that
\begin{equation}
\mathcal{F}
\equiv\left\{v\,\big|\, v^0>0\,,\, \nabla\cdot v=0\,,\, |v|\geq 1\right\}\,.
\end{equation}
Following \cite{Pedraza:2021mkh,Pedraza:2021fgp}, we can now express holographic complexity as the minimum flux through the region $A$, in units of the arbitrary scale $G_N \ell$,
\begin{equation}
\begin{split}
    \mathcal{C}(\sigma)=\frac{1}{G_N\ell}\min_{v\in \mathcal{F}} \int_{A}v\,,
    \label{flowsprogram}
\end{split}
\end{equation}
with $\int_{A}v=\int_{A}d^{d-1}x\sqrt{|h|}n_\mu v^\mu$. Here, $h$ is the determinant of the induced metric on $\partial\mathcal{M}$ and $n^\mu$ its future directed unit normal vector. The min flow-max cut theorem \cite{Headrick:2017ucz} implies that the optimal solution to this optimization program is equal to the maximal volume of a surface homologous to the region $A$, and thus, (\ref{flowsprogram}) is formally equivalent to the CV proposal (\ref{eq:CVduality}).

As discussed in \cite{Pedraza:2021mkh,Pedraza:2021fgp}, it is useful to discretize the flow program and interpret it in terms of a collection of microscopic Lorentzian threads. This approach parallels the Riemannian `bit threads' program for holographic entanglement entropy \cite{Freedman:2016zud}, which has been extensively used to understand entropy inequalities \cite{Cui:2018dyq,Hubeny:2018bri,Agon:2018lwq} and many other properties of holographic entanglement \cite{Harper:2018sdd,Harper:2019lff,Du:2019emy,Bao:2019wcf,Agon:2019qgh,Agon:2020mvu,Lin:2020yzf,Headrick:2020gyq,Agon:2021tia,Rolph:2021hgz,Mintchev:2022fcp,Gursoy:2023tdx,Caggioli:2024uza,Du:2024xoz,Headrick:2025awv,Wu:2025qwc,Das:2025fav}. These Lorentzian threads correspond to the integral lines of the flow that optimizes \eqref{flowsprogram}, and are characterized by transverse density $\rho=|v|/G_N\ell$. In this discrete version, complexity is determined by the minimal number of threads ending in $A$, $N_A$, which equals the volume of the maximal bulk slice $\Sigma(A)$ homologous to $A$:
\begin{equation}
    \mathcal{C}(\sigma)=\min N_A=\frac{1}{G_N\ell}\max [\textup{Vol}(\Sigma(A))]\,.
\end{equation}
From a more physical perspective, Lorentzian threads can be interpreted as a set of unitary operators (gates) that prepare a state on $\Sigma(A)$, starting from a CFT state in the infinite past of $\mathcal{M}$. In this picture, complexity is given by the minimal number of operations required to go from the reference state to the target state, closely mirroring the concept of computational (circuit) complexity in quantum information theory.

It is important to clarify a subtlety regarding the above interpretation. The proposal in \cite{Pedraza:2021mkh,Pedraza:2021fgp} assumed that the characteristic thickness of Lorentzian threads is $\ell_P^{d-1}$, where  $\ell_P$ is the Planck length and $d$ is the number of spacetime dimensions. This requires that the arbitrary length in the CV proposal should be identified with the Planck length, $\ell\sim\ell_P$. Consequently, the number of gates required to implement bulk evolution aligns with the number of physical degrees of freedom in the smallest tensor network anchored to $\sigma$, as suggested in \cite{Stanford:2014jda}.

However, here we argue that the scale $\ell$ must necessarily be larger than $\ell_P$, and should therefore be interpreted as a coarse-graining scale. To see this, notice that each thread should be capable of connecting multiple degrees of freedom simultaneously. This is necessary because time evolution in a general quantum system must enable the generation of entanglement among different degrees of freedom, and this is only possible if the universal set of gates contains elementary operations that act upon multiple degrees of freedom at once. We define the ratio $\lfloor \ell/\ell_P\rfloor=k$ as the level of $k$-locality of the quantum circuit, as illustrated in Fig.~\ref{tensorNetwork}.
\begin{figure}
    \centering
    \includegraphics[width=0.2\textwidth]{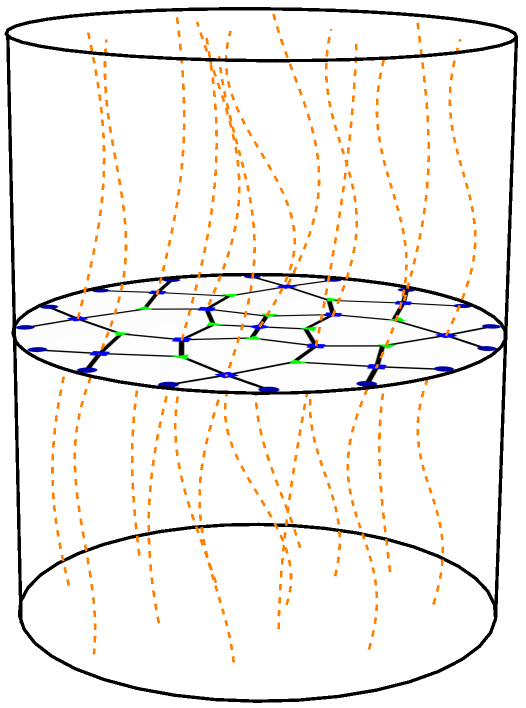}\includegraphics[width=0.045\textwidth,trim={0 -0.85cm 0 0},clip]{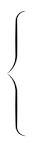}\hspace{-0.33cm}\includegraphics[width=0.23\textwidth,trim={0 -3.4cm 0 0},clip]{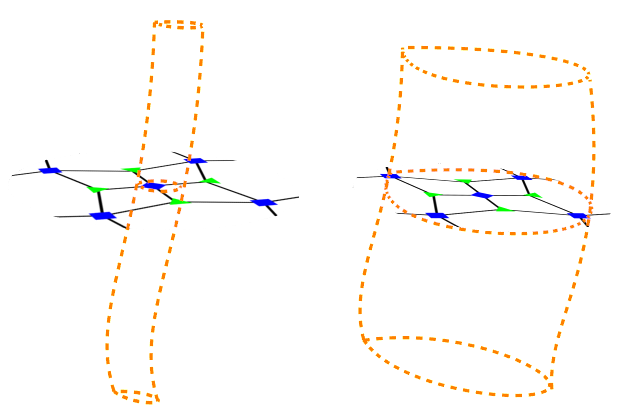}
    \put(-102,20){$k=1$}
    \put(-41,20){$k=5$}
    \caption{Discretization of the bulk slice $\Sigma(A)$ using a physical tensor network. On the right, two threads of different thicknesses are illustrated. The number of degrees of freedom these threads act upon represents the level of $k$-locality in the bulk quantum circuit.}
    \label{tensorNetwork}
\end{figure}

\subsection{Prescription in terms of measures}\label{prescription in terms of measures}
While the flow prescription is valid and useful in scenarios requiring numerical results, for the purposes of this work, it is more practical to focus on the discretized Lorentzian threads approach. Following \cite{Headrick:2022nbe}, we will now formulate the program in the language of measure theory. We refer the reader to Appendix \ref{App: Riemannian threads and hyperthreads}, where we review the Riemannian threads program in this language.

Let $\mathcal{P}$ be the family of subsets of a set $\Omega$.\footnote{For $\mathcal{P}$ to be well-defined, it must include the empty set $\emptyset$ and $\Omega$ itself, and be closed under union and difference of subsets.} A measure is a function $\mu:\mathcal{P}\rightarrow [-\infty,\infty]$ that acts trivially on the empty set, $\mu(\emptyset)=0$ and satisfies $\mu( \cup_{n}\mathcal{P}_n)=\sum_{n}\mu(\mathcal{P}_n)$ for any sequence of disjoint sets $\mathcal{P}_n\in\mathcal{P}$ \cite{rao1983theory}. In our context, $\mathcal{P}$ consists of all timelike, future-directed curves within $\mathcal{M}$ that start in $\partial \mathcal{M}\backslash A$ and end in $A$ ---referred to as Lorentzian threads or simply threads. We will focus on measures $\mu$ that assign each thread a value in $\{0,1\}$. 

To write the program in terms of measures, we need to express the flow conditions using the function $\mu$ \cite{Harper:2021uuq}. The divergenceless condition is naturally satisfied, as the threads terminate on the boundary rather than in the bulk. To incorporate the norm bound, we introduce a delta function $\Delta(x,p)$ defined as:
\begin{equation}
    \Delta(x,p)=\int_p ds \delta(x-y(s))\,,
    \label{deltafunction}
\end{equation}
where $s$ is an affine parameter, and $y(s)$ represents the trajectory of the thread $p$. This function counts the number of times a thread passes through a point $x$, enabling the introduction of a local thread density $\rho(x)$:
\begin{equation}
    \rho(x)=\int_{\mathcal{P}} d\mu \Delta(x,p)\,.
    \label{density}
\end{equation}
Note that the integral over the measure offers a new way to express the sum of the weights of all threads in $\mathcal{P}$. Mathematically, this is written as $\int_{\mathcal{P}}d\mu\equiv \sum_{p\in \mathcal{P}}\mu(p)$. In this representation, the norm bound is expressed as:\footnote{To streamline the notation, we will use units where $G_N\ell = 1$ from this point forward, unless otherwise specified.}
\begin{equation}
    \rho(x)\geq 1\,,\ \ \forall x\in \mathcal{M}\,.
\end{equation}
Finally, we require that any macroscopically small region on the boundary contains a non-zero number of threads. Viewing threads as the integral curves of a bulk vector field, this condition ensures a non-vanishing flux through the boundary and hence that the solution remains homologous to $A$, as we will clarify in more detail below.

Having established all the essential definitions, we now propose the following objective function in the space of measures that yields complexity through minimization:
\begin{equation}
    \mathcal{C}(\sigma)=\min \int_{\mathcal{P}}d\mu\,, \ \textup{s.t. } \rho(x)\geq 1\ \ \forall x\in \mathcal{M}\,.
\label{Program one party complexity}
\end{equation}
To tackle this optimization problem, we formulate a Lagrangian based on the objective function, incorporating a Lagrange multiplier $\lambda$ to enforce the norm bound. As discussed in \cite{boyd2004convex}, the Lagrange multiplier associated with an inequality must always be non-negative, which in our case implies that $\lambda(x)\geq 0\ \forall x\in \mathcal{M}$. Consequently, the primal Lagrangian can be expressed as follows:
\begin{equation}
\begin{split}
    L(\mu,\lambda)=&\int_{\mathcal{P}}d\mu +\int_{\mathcal{M}}\lambda(x)\left(1-\int_{\mathcal{P}}d\mu\Delta(x,p)\right)\\
    =&\int_{\mathcal{M}} \lambda(x)+\int_{\mathcal{P}}d\mu\left(1-\int_{p}ds\lambda(x)\right)\,.
\end{split}
\end{equation}
Note that in the integrals over $\mathcal{M}$, we omit the volume form $d^dx\sqrt{-g}$ for simplicity. The second line introduces a new program that is dual to the original one:
\begin{equation}
    \max \int_{\mathcal{M}}d^dx \sqrt{-g}\lambda(x)\,,\ \textup{s.t.}\int_{p}ds\lambda(x)\leq 1 \ \ \forall p\in\mathcal{P}\,,
    \label{eq: dualprogramcomplexity}
\end{equation}
where integration over $p$ entails computing the integral along the path traced by the thread $p$. Before computing the optimal value of the dual program, it is essential to confirm that its solution matches that of the primal program, a requirement known as strong duality. This is satisfied if Slater's condition holds, which stipulates the existence of a solution (not necessarily optimal) that strictly meets the inequality constraints \cite{Slater2014}. While identifying this solution for the primal program is difficult, for the dual program, setting $\lambda=0 \ \ \forall x\in\mathcal{M}$ yields $\int_p ds \lambda =0\ \ \forall\ p \in\mathcal{P}$. This verifies strong duality, allowing us to confidently proceed with the dual program.

Let us continue our search for the optimal solution to the program, beginning with a heuristic approach followed by a more rigorous analysis. In \cite{Harper:2021uuq}, it was suggested that $\lambda(x)$ can be better understood through its collection of level sets. Each thread crosses a certain number of these surfaces, and the sum of the values of the level sets crossed by each thread must be less than or equal to 1. An optimal solution may be achieved when these level sets are densely packed in a hypersurface of maximal volume $\tilde{\Sigma}(A)$ homologous to $A$ \emph{relative to the boundary}.\footnote{See \cite{Headrick:2017ucz} for a discussion on the distinction between homology and relative homology.} However, if this is true, none of the threads originating in $ J^+(\tilde{\Sigma}(A))\cap \partial\mathcal{M}$ would cross such a surface. Consequently, due to complementary slackness, their weight would be zero, as they do not satisfy the constraint. Unless the surface is homologous to $A$, there would be no flux through part of the boundary, $(J^+(\tilde{\Sigma}(A))\cap \partial\mathcal{M})\backslash A$, contradicting our initial assumption. 
This ensures that every thread in $\mathcal{P}$ crosses the barrier and, therefore, contributes to the flux through region $A$. Furthermore, it is clear that a $\lambda$ with support only on the maximal surface $\Sigma\sim A$ maximizes the objective; any variation that adheres to the density bound would decrease the objective value.

From a mathematical standpoint, we can provide a rigorous solution to the problem instead of relying solely on a pictorial argument. The proof is based on the following theorem (proved in \cite{Caceres:2023ziv} and revisited in Appendix \ref{app: proof}):
\begin{restatable}{theorem}{thm}\label{th: theorem1}
Let $\mathcal{M}$ be a Lorentzian manifold, and let $A$ and $B$ be complementary subsets of the boundary such that $J^+(A) \cap \partial \mathcal{M} = A$ and $J^-(B) \cap \partial \mathcal{M} = B$. Define $\mathcal{P}$ as the set of timelike, future-directed (FD) curves from $B$ to $A$, and let $\lambda(x)$ be a non-negative function on $\mathcal{M}$. The following two statements are then equivalent:
    \begin{equation}
    \begin{split}
        &\exists \psi:\mathcal{M}\rightarrow [-1/2,1/2]\ \ s.t.\ \ \psi|_B=-1/2\,,\\ 
        &\psi|_A=1/2\,,\ |d\psi|\geq \lambda\,,\ d\psi\ \text{timelike \& FD}\,,
        \label{eq: condition1}
    \end{split}
    \end{equation}
and
    \begin{equation}
        \forall p\in \mathcal{P},\quad \int_p ds\lambda\leq 1\,.
        \label{eq: condition2}
    \end{equation}
\end{restatable}

From the conditions in \eqref{eq: condition1}, we can assert that
\begin{equation}
    \max \int_{\mathcal{M}}\lambda \leq \int_\mathcal{M} |d\psi|\,.
\end{equation}
Following \cite{Headrick:2017ucz}, we now define the following set of regions:
\begin{equation}
    r(p)=\{x\in\mathcal{M}|\psi(x)\geq p\}\,,
\end{equation}
and its closure in the bulk, denoted by $\Sigma(p)=\partial r(p)\backslash \partial M$. Since $d\psi$ is timelike and future-directed, each $\Sigma(p)$ represents a Cauchy slice. Furthermore, since $\psi|_B = -1/2$ and $\psi|_A = 1/2$, these slices are homologous to $A$. Applying the `co-area formula' in \cite{Treude:2012np} one finds
\begin{equation}
        \int_{\mathcal{M}}|d\psi|=\int_{-1/2}^{1/2} dp \textup{Vol}(\Sigma(p))\leq \max_{\Sigma(p)\sim A}\textup{Vol}(\Sigma(p))\,.
\end{equation}
This inequality indicates that the optimal solution is bounded above by the volume of the maximal hypersurface homologous to $A$. Demonstrating the equality of this bound is straightforward: the volume of this surface can be represented mathematically as the integral of a delta function over $\mathcal{M}$, with support on $\Sigma$. It is clear that the inequality in \eqref{eq: dualprogramcomplexity} is satisfied, leading to $\max \int_{\mathcal{M}}\lambda\geq \max_{\Sigma(p)\sim A}\textup{Vol}(\Sigma(p))$. Therefore,
\begin{equation}
   \mathcal{C} =\min_{\mu}\int_{\mathcal{P}}d\mu=\max_{\Sigma(p)\sim A}\textup{Vol}(\Sigma(p))\,.
\end{equation}
This completes our reformulation of Lorentzian threads in the language of measure theory.

\subsection{Shortcomings}\label{shortcomings}
In the previous subsection, we showed that, under the CV prescription, the complexity of holographic states can be determined by solving an optimization problem that minimizes the number of threads satisfying the density bound $\rho(x)\geq 1$. This raises a natural question: can this formalism also be applied to compute subregion complexity? Specifically, if the boundary slice $\sigma(A)=\partial A$ is divided into several subregions $\sigma_i(A)$, $i\in\{1,\ldots,N\}$, can the complexity of each subregion be determined simply by counting the number of threads in a given configuration? In the Riemannian case, the bit thread program effectively computes the entanglement entropy of arbitrary subregions, a result guaranteed by the `multicommodity' theorem \cite{Cui:2018dyq}. Specifically, this theorem guarantees an optimal solution represented by the sum of the entropies of all subregions.  However, as we will show below, this approach is insufficient in the Lorentzian case.

To illustrate the point, we will focus on a bipartite system with two boundary regions, $\sigma(A)=\sigma_1(A)\cup\sigma_2(A)$, and using the same methodology, we will attempt to compute $\mathcal{C}(\sigma_1(A))$ and $\mathcal{C}(\sigma_2(A))$ simultaneously. According to the CV proposal, the complexity of a subregion $\sigma_i(A)$ is given by the volume of the maximal Cauchy surface $\Sigma_i(A)$, which is bounded by $\sigma_i(A)$ and the associated HRT surface $\gamma_i$ \cite{Agon:2018zso,Caceres:2018blh}. For a pure overall state in the bipartite case, we find that $\gamma_1=\gamma_2=\gamma$. Additionally, it is crucial to construct the entanglement wedges of the corresponding subregions, defined as the domains of dependence $D_i(A)$ of generic slices $\Sigma_i(A)$ such that $\partial \Sigma_i(A)=\sigma_i(A) \cup \gamma_i$. See Fig. \ref{fig: bipartite} for an illustration.
\begin{figure}
    \centering
    \includegraphics[width=0.4\textwidth]{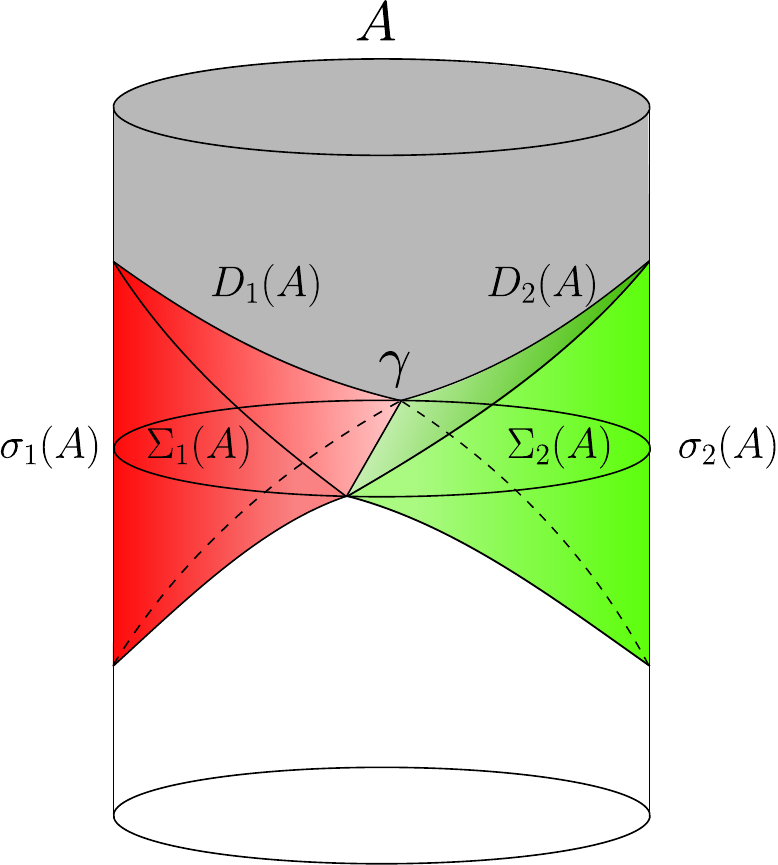}
    \caption{Illustration of a bipartite system: the boundary Cauchy slice $\sigma(A)$ is divided into two regions, $\sigma_1(A)$ and $\sigma_2(A)$, which share a common HRT surface, $\gamma$. The surfaces $\Sigma_1(A)$ and $\Sigma_2(A)$ are bounded by $\sigma_1(A) \cup \gamma$ and $\sigma_2(A) \cup \gamma$, respectively, with their corresponding domains of dependence, $D_1(A)$ and $D_2(A)$, depicted in red and green.}
    \label{fig: bipartite}
\end{figure}

We will now introduce the following subsets of $\mathcal{P}$:
\begin{equation}
    \begin{split}
        \mathcal{P}&=\{\textup{All threads}\}\,,\\
        \mathcal{P}_{1}&=\{\textup{Threads crossing }D_1(A)\}\,,\\
        \mathcal{P}_{2}&=\{\textup{Threads crossing }D_2(A)\}\,.
    \end{split}
\end{equation}
Assuming that a version of the `multicommodity' theorem extends to the Lorentzian case, one could hypothesize the existence of a global solution where $\mathcal{C}(\sigma_i(A))$ is determined by counting the threads passing through the corresponding domains of dependence $D_i(A)$. Following this reasoning, the program for calculating the subregion complexities of both $\sigma_1(A)$ and $\sigma_2(A)$ would be:
\begin{equation}
    \min \left[\int_{\mathcal{P}_1}d\mu+\int_{\mathcal{P}_2}d\mu\right], \ \textup{s.t. } \rho(x)\geq 1\ \ \forall x\in \mathcal{M}\,,
\end{equation}
whose Lagrangian is:
\begin{equation}
\begin{split}
    L=&\int_{\mathcal{P}_1}d\mu+\int_{\mathcal{P}_2}d\mu+\int_{\mathcal{M}}\lambda(x)\left(1-\rho(x)\right)\\
    =&\int_{\mathcal{P}}d\mu+\int_{\mathcal{M}}\lambda(x)\left(1-\int_{\mathcal{P}}d\mu\Delta(x,p)\right).
\end{split}
\end{equation}
As we can see, this program matches with that presented in \eqref{Program one party complexity}, leading to the conclusion that $\mathcal{C}(\sigma_1(A))+\mathcal{C}(\sigma_2(A))=\mathcal{C}(\sigma(A))$. However, this results in a contradiction: the CV proposal asserts that subregion complexity is superadditive \cite{Agon:2018zso,Caceres:2018blh}, implying that the `mutual complexity' $\Delta \mathcal{C}$ must satisfy \cite{Caceres:2019pgf}
\begin{equation}\label{mutualcomplexity}
\Delta \mathcal{C}\equiv \mathcal{C}(\sigma_1(A))+\mathcal{C}(\sigma_2(A))-\mathcal{C}(\sigma(A))\leq 0\,,
\end{equation}
though it does not necessarily vanish in general. This contradiction reveals a flaw in the current formulation of Lorentzian threads while offering valuable insights for developing a more general approach. Specifically, the new framework should be robust enough to allow for the simultaneous calculation of the complexities of all subregions and their union. The development of this formulation will be the primary focus of the next section.

\section{Multiple thread flavors}\label{sec: multiple flavors}
As outlined in the previous section, our focus now shifts to developing a more general formalism capable of calculating both subregion complexities and the complexity of their union. In the Riemannian case, this is achievable due to the multicommodity theorem. For instance, in a tripartite system where the boundary is divided into three regions ---$A$, $B$, and $C=(A\cup B)^c$, analogous to the scenario considered in Sec.~\ref{shortcomings}--- threads can connect $A$ to $B$, $A$ to $C$, and $B$ to $C$. However, in the Lorentzian case, the timelike nature of Lorentzian threads prevents any threads from connecting $D_i(A)$ to $D_j(A)$ for $i \neq j$.

To resolve this issue, we will introduce a generalized measure that distinguishes between these regions while adhering to the previously established conditions. We will begin by applying this approach to a bipartite system for simplicity, and then extend it to an $N$-partite system, thereby enabling the calculation of complexities across systems with an arbitrary number of subregions.

\subsection{Bipartite system}

In this section, we build on the setup described in Sec.~\ref{shortcomings}. We will start with a brief recap of its key elements for clarity. We consider a Cauchy slice $\sigma(A)$ in the boundary CFT, which is divided into two disjoint subregions, $\sigma_1(A)$ and $\sigma_2(A)$, which share a common boundary, $\partial \sigma_1(A) = \partial \sigma_2(A)$. We also introduce the HRT surface $\gamma$, shared by both subregions, and construct the corresponding entanglement wedges, or domains of dependence, $D_1(A)$ and $D_2(A)$.

At this point, we introduce two distinct measures, $\mu_1$ and $\mu_2$, corresponding to two different `flavors' of Lorentzian threads. The first measure, $\mu_1$, will be used exclusively for computing complexities related to $\sigma_1(A)$, such as $\mathcal{C}(\sigma_1(A))$ and $\mathcal{C}(\sigma(A))$, but not $\mathcal{C}(\sigma_2(A))$. Similarly, $\mu_2$ will be employed for computations involving $\sigma_2(A)$. A visual representation of this setup is shown in Fig.~\ref{fig:bipartite thread}. The program we propose for simultaneously calculating $\mathcal{C}(\sigma_1(A))$, $\mathcal{C}(\sigma_2(A))$, and $\mathcal{C}(\sigma(A))$ is as follows:
\begin{equation}
\begin{split}
    \min &\left[\int_{\mathcal{P}_1}d\mu_1+\int_{\mathcal{P}_2}d\mu_2+\frac{1}{2}\int_{\mathcal{P}}(d\mu_1+d\mu_2)\right]\\
    &\textup{s.t. }\rho_1\geq 1,\ \rho_2\geq 1\ \ \forall x\in \mathcal{M}\,,
\end{split}
\label{programbipartite}
\end{equation}
with $\rho_1=\int_{\mathcal{P}}d\mu_1 \Delta(x,p)$ and $\rho_2=\int_{\mathcal{P}}d\mu_2 \Delta(x,p)$. The presence of the $1/2$ factor in the last term will become clear in the discussion below.
\begin{figure}[t!]
    \centering
    \includegraphics[width=0.35\textwidth]{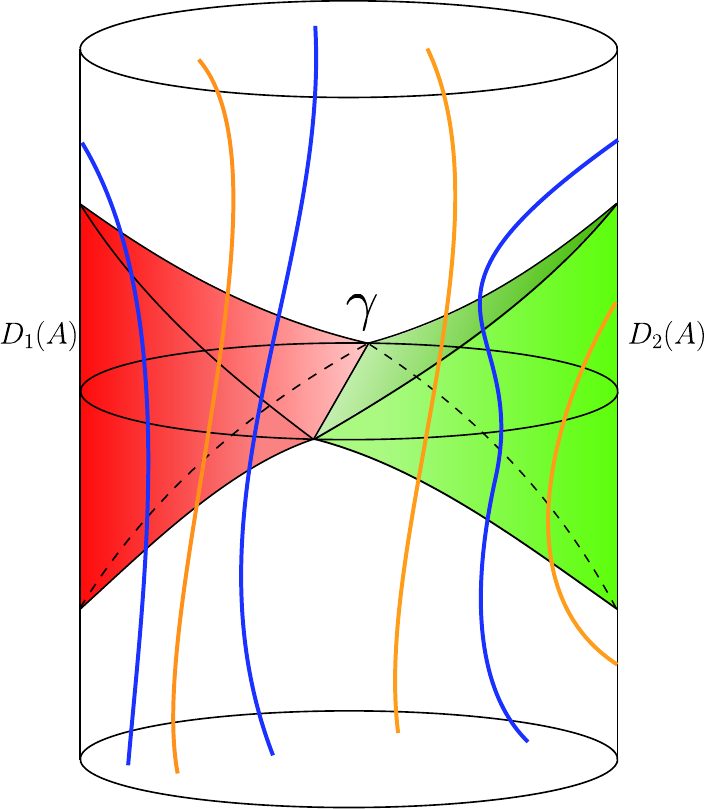}
    \caption{Threads in the manifold $\mathcal{M}$ with measures $\mu_{1}$ or $\mu_{2}$, depicted in orange and blue, respectively. In this scenario, the complexity of $\sigma_1(A)$ is determined by the number of orange threads traversing $D_1(A)$, while the complexity of $\sigma_2(A)$ corresponds to the number of blue threads passing through $D_2(A)$. The complexity of $\sigma(A)$ is calculated as half of the total number of threads present, regardless of their color.}
    \label{fig:bipartite thread}
\end{figure}

As in the previous section, we require that for any macroscopic but small area on $\partial\mathcal{M}$, a set of threads of both types is attached. The resulting Lagrangian is then:
\begin{align}
        L=&\int_{\mathcal{P}_1}d\mu_1+\int_{\mathcal{P}_2}d\mu_2+\frac{1}{2}\int_{\mathcal{P}}(d\mu_1+d\mu_2)\notag\\
        &+\int_{\mathcal{M}}d^dx \sqrt{-g}\left[\lambda_1\left(1-\int_{\mathcal{P}}d\mu_A\Delta(x,p)\right)\right.\notag\\
        &\left.+\lambda_2\left(1-\int_{\mathcal{P}}d\mu_2\Delta(x,p)\right)\right]\\
        =&\int_{\mathcal{M}}d^d x\sqrt{-g}\lambda_1+\int_{\mathcal{P}_1}d\mu_1 \left[3/2-\int_{p}ds\lambda_1\right]\notag\\
        &+\int_{\mathcal{P}_2}d\mu_1 \left[1/2-\int_{p}ds\lambda_1\right]+(1\leftrightarrow 2)\,.\notag
\end{align}
We observe that the program can be decomposed into the sum of two simpler programs: one dependent solely on $\lambda_1$ and $\mu_1$, and the other on $\lambda_2$ and $\mu_2$. Consequently, the dual program,
\begin{equation}
\begin{split}
    \max &\int_{\mathcal{M}} d^d x \sqrt{-g}(\lambda_1+\lambda_2)\ \ \textup{s.t. }\\
  &\int_{p}ds \lambda_{1}\leq 3/2\ \forall\ p\in \mathcal{P}_1\ \&\ p\in \textup{dom}(\mu_1)\,, \\
  &\int_{p}ds \lambda_{1}\leq 1/2\ \forall\ p\in \mathcal{P}_2\ \&\ p\in \textup{dom}(\mu_1)\,,\\
  &\int_{p}ds \lambda_{2}\leq 3/2\ \forall\ p\in \mathcal{P}_2\ \&\ p\in \textup{dom}(\mu_2)\,,\\
  &\int_{p}ds \lambda_{2}\leq 1/2\ \forall\ p\in \mathcal{P}_1\ \&\ p\in \textup{dom}(\mu_2)\,,
 \end{split}
 \label{eq: completeprogram}
\end{equation}
can likewise be expressed as the sum of two simpler ones:
\begin{equation}
\begin{split}
 \max &\int_{\mathcal{M}} d^d x \sqrt{-g}\lambda_1\ \  \textup{s.t. } \\
 &\int_{p}ds \lambda_{1}\leq 3/2\ \forall\ p\in \mathcal{P}_1 \ \&\  p\in \textup{dom}(\mu_1)\,,\\
 &\int_{p}ds \lambda_{1}\leq 1/2\ \forall\ p\in \mathcal{P}_2  \ \&\ p\in \textup{dom}(\mu_1)\,,
 \end{split}
\label{eq: separated program A}
\end{equation}
and
\begin{equation} 
\begin{split}
\max &\int_{\mathcal{M}} d^d x \sqrt{-g}\lambda_2\ \ \textup{s.t. } \\
&\int_{p}ds \lambda_{2}\leq 3/2\ \forall\ p\in \mathcal{P}_2\ \&\ p\in \textup{dom}(\mu_2)\,,\\
&\int_{p}ds \lambda_{2}\leq 1/2\ \forall\ p\in \mathcal{P}_1\ \&\ p\in \textup{dom}(\mu_2)\,.
 \end{split}
 \label{eq: separated program B}
\end{equation}

Following a similar reasoning as in Sec.~\ref{prescription in terms of measures}, as long as the flux condition through the boundary is satisfied, the solution to the first program is the sum of the maximal volume of a surface anchored to $\sigma_1(A) \cup \gamma$ and half the maximal volume of a surface anchored to $\sigma(A)$. Similarly, the solution to the second program is the sum of the maximal volume of a surface anchored to $\sigma_2(A) \cup \gamma$ and half the maximal volume of a surface anchored to $\sigma(A)$. Therefore, the solution to the complete program \eqref{eq: completeprogram} is:
\begin{equation}\label{eq: bipartitesolution}
    \begin{split}
    \mathcal{C}(\sigma_1(A))+\mathcal{C}(\sigma_2(A))+\mathcal{C}(\sigma(A))=\ \ \\
    \max \left[\textup{Vol}(\Sigma_1(A))\right]+\max \left[\textup{Vol}(\Sigma_2(A))\right]\\
    +\max \left[\textup{Vol}(\Sigma(A))\right]\,.\qquad\quad \ \,
    \end{split}
\end{equation}
This result clarifies the inclusion of the $1/2$ factor in the multi-flavor program (\ref{programbipartite}); without this factor, the optimal solution would yield $\mathcal{C}(\sigma_1(A)) + \mathcal{C}(\sigma_2(A)) + 2\mathcal{C}(\sigma(A))$. To conclude this section, it is important to emphasize that since $\Sigma(A)$ is not constrained to cross the HRT surface, its volume will always be bounded from below by the sum of the volumes of $\Sigma_1(A)$ and $\Sigma_2(A)$, thereby ensuring that superadditivity \cite{Agon:2018zso} is maintained.

We have just provided a heuristic explanation for \eqref{eq: bipartitesolution}, and we will now present a more formal proof. To do so, we assume that for each program, $\lambda_i$ ($i=1,2$) has support either in $\mathcal{M}^-=D_1(A)\cup D_2(A)\cup J^-(D_1(A)\cup D_2(A))$ or in $\mathcal{M}^+=D_1(A)\cup D_2(A)\cup J^+(D_1(A)\cup D_2(A))$. This assumption is physically sensible, as any violation would imply that the optimal solution is not spacelike everywhere. We will now first consider the case where the support of such a function is contained within $\mathcal{M}^-$ and then move to the second case. To do so, we will use the following theorem, proved in Appendix \ref{app: proof}:

\begin{restatable}{theorem}{segundo}\label{th: theorem2}
    Let $\mathcal{M}$ be a Lorentzian manifold and $\sigma(A)$ a boundary Cauchy slice divided into two complementary regions, $\sigma_1(A)$ and $\sigma_2(A)$. Let $D_1(A)$ and $D_2(A)$ denote the domains of dependence of any bulk Cauchy slices anchored to $\sigma_1(A)\cup \gamma$ and $\sigma_2(A)\cup \gamma$, respectively, where $\gamma$ is the HRT surface associated with $\sigma_1(A)$ and $\sigma_2(A)$. Define $\mathcal{M}^-$ as the union of $D_1(A)$, $D_2(A)$ and $J^-(D_1(A)\cup D_2(A))$. Let $\partial D^+_1$ represent the future boundary of $D_1(A)$, and $\partial D^+_2$ that of $D_2(A)$, and denote the intersection of the boundary and the causal past of $\sigma_A\cup \sigma_B$ by $C$. Finally, let $\mathcal{P}_1$ be the set of all threads crossing $D_1(A)$ and $\mathcal{P}_2$ the set of all threads passing through $D_2(A)$. The following two statements are equivalent:
    \begin{equation}
    \begin{split}
        &\exists \psi:\mathcal{M}^-\rightarrow [-1/4,5/4]\ \ s.t.\ \psi|_C=-1/4\,,\\
        &\psi|_{\partial D_1^+}=5/4,\ \psi|_{\partial D_2^+}=1/4\,,\\
        &|d\psi|\geq \lambda,\ d\psi\ \text{timelike \& FD}\,,
        \label{eq: condition11}
    \end{split}
    \end{equation}
and
    \begin{equation}
    \begin{split}
        \forall p_1\in \mathcal{P}_1,\quad \int_{p_1} ds\lambda\leq 3/2\,,\\
        \forall p_2\in \mathcal{P}_2,\quad \int_{p_2} ds\lambda\leq 1/2\,,
        \label{eq: condition22}
    \end{split}
    \end{equation}
where $s$ is the proper distance along $p$.
\end{restatable}

The above theorem guarantees the existence of a function defined on $\mathcal{M}^-$ that is everywhere greater than $\lambda_i$. This implies,
\begin{equation}
    \begin{split}
    &\int_{\mathcal{M}}\lambda \leq  \int_{\mathcal{M}} |d\psi|=\int_{-\infty}^{\infty} dp \Vol(m(p)) \\
    &\leq \int_{-1/4}^{1/4} dp \Vol(m(p))+\int_{1/4}^{5/4} dp \Vol(m(p))\,.
    \end{split}
    \label{eq: nesting}
\end{equation}
Since $\psi$ is constant over $C$ and takes different values over $\partial D_1^+$ and $\partial D_2^+$, all level sets $m(p)$ with $p\in(-1/4,1/4)$ have $\sigma(A)$ as their boundary, while those with $p\in(1/4,5/4)$ are anchored to $\sigma_1(A)\cup \gamma$. It is straightforward to see that the first term on the right-hand side of \eqref{eq: nesting} is bounded above by half the volume of the maximal volume slice fully contained in $\mathcal{M}^-$ homologous to $\sigma(A)$, while the second term is bounded by the volume of the maximal slice anchored to $\sigma_1(A)\cup \gamma$. A similar procedure can be applied when assuming that the optimal solution is found within $\mathcal{M}^+$, yielding an analogous result. One would then select the maximal solution between these two cases. Consequently, we conclude that, under our assumptions, the solution to \eqref{eq: separated program A} is bounded above by $\max \left[\textup{Vol}(\Sigma_1(A))\right]+\frac{1}{2}\max \left[\textup{Vol}(\Sigma(A))\right]$. 

On the other hand, provided that half a delta function supported on the maximal surface homologous to $\sigma(A)$ plus another delta with a non-vanishing value on the maximal surface bounded by $\sigma_1(A)\cup \gamma$ satisfies the conditions of the program, we can guarantee that this solution is bounded below by the same value. Combining both bounds, we can conclude that the optimal solution to the program is exactly $\max \left[\textup{Vol}(\Sigma_1(A))\right]+\frac{1}{2}\max \left[\textup{Vol}(\Sigma(A))\right]$.  Applying the same reasoning to \eqref{eq: separated program B} leads us to a full solution to \eqref{eq: completeprogram}, which, as expected, matches \eqref{eq: bipartitesolution}.

Before finishing this subsection, let us briefly justify the central assumption around which our proof is constructed. As the reader may have noticed, the derivation provided here closely resembles the one in \cite{Headrick:2017ucz} used to prove the nesting property for entanglement entropy. A key point in that derivation is the fact that the value of the function $\psi$ on the boundary is known. In our case, this is not true: there is no natural division of the boundary region $J^+(\sigma(A))$ into two pieces where $\psi=5/4$ and $\psi=1/4$ respectively such that their boundaries are just $\sigma_1(A)$ and $\sigma_2(A)$. This prevents the existence of a globally defined function $\psi$ as described in Theorem \ref{th: theorem2}. For this reason, it is convenient to restrict the domain of this function to either $\mathcal{M}^-$ or $\mathcal{M}^+$. Nonetheless, all known solutions to problems of this kind (for instance, the flow-cut theorems or the nesting property \cite{Headrick:2017ucz}) are compatible with a function $\lambda$ whose kernel is given by either $\mathcal{M}\setminus\mathcal{M}^-$ or $\mathcal{M}\setminus\mathcal{M}^+$. Accordingly, the assumption about the support of the function $\lambda$ is well justified.

\subsection{Multipartite generalization}\label{multipartite generalization}
Let us now consider the case of $N$ boundary regions $\sigma_i(A)$ with $i=1,...,N$. We assume that $\displaystyle\cup_{i=1}^N\sigma_i(A)=\sigma(A)$ and $\sigma_i(A)\cap \sigma_j(A)=\emptyset$ for $i\neq j$. We note that, for $N\geq 3$, an arbitrary bulk Cauchy slice $\Sigma$, such that $\partial \Sigma=\sigma(A)$ and which contains all the HRT surfaces associated with each $\sigma_i(A)$, will not be fully contained within the union of domains of dependence $\cup_{i=1}^N D_i(A)$. The complement defines a central region, dubbed the \emph{hole}, which plays a crucial role in quantum error-correction studies of AdS/CFT \cite{Almheiri:2014lwa}. This region is the bulk locus whose information is encoded non-locally on the boundary; in short, it embodies the secret-sharing (or error-correcting) feature of AdS/CFT by storing logical qubits that can survive the erasure of any single boundary region while remaining inaccessible to that region alone.

Because of the hole, we must now introduce a new set of threads that was not needed previously:
\begin{equation}
\begin{split}
    \mathcal{P}&=\{\textup{All threads}\}\,,\\
    \mathcal{P}_{i}&=\{\textup{Threads passing through } D_i(A)\}\,,\\
    \mathcal{P}_{N+1}&=\{\textup{Threads passing through the hole} \}\,.
\end{split}
\end{equation}
Following the logic of the previous subsection, we now propose a program that computes the complexity of each subregion $\sigma_i(A)$ and the complexity of their union,
\begin{equation}
\begin{split}
    \min &\sum_{i=1}^N\left(\int_{\mathcal{P}_i}d\mu_i+\frac{1}{N}\int_{\mathcal{P}} d\mu_i\right)\\
    &\textup{s.t. }\rho_i\geq 1\ \forall i=1,...,N\,.
\end{split}
\label{lagrangianmultipleregions}
\end{equation}
We keep implicit the boundary condition on thread flux, namely, that for any macroscopic, yet small area element of $\partial\mathcal{M}$, a set of threads of each type attaches to the boundary. The corresponding Lagrangian is given by:
\begin{equation}
    \begin{split}
        L=&\displaystyle \sum_{i=1}^N\int_{{\mathcal{P}}_i}d\mu_i+\frac{1}{N}\int_{\mathcal{P}} \sum_{i=1}^N d\mu_i\\
        &+\int_{\mathcal{M}} \sum_{i=1}^N\left(\lambda_i-\int_{\mathcal{P}} d\mu_i\lambda_i \Delta(x,p)\right)\\
        =&\sum_{i=1}^N\left[\int_{\mathcal{M}} \lambda_i+\int_{{\mathcal{P}}_i}d\mu_i\left(\frac{N+1}{N}-\int_p ds \lambda_i\right)\right.\\
        &+\sum_{j\neq i}\int_{{\mathcal{P}}_j}d\mu_i\left(\frac{1}{N}-\int_p ds \lambda_i\right)\\
        &+\left.\int_{{\mathcal{P}}_{N+1}}d\mu_i\left(\frac{1}{N}-\int_p ds \lambda_i\right)\right].
    \end{split}
    \label{programmultipartite}
\end{equation}
Similarly, the dual program decomposes into $N$ independent subprograms, which can be written as:
\begin{equation}
\begin{split}
    \max& \int_{\mathcal{M}} d^dx \sqrt{-g}\displaystyle \sum_{i=1}^N \lambda_i \quad \textup{ s.t. }\\
    & \int_p ds \lambda_i\leq \frac{N+1}{N},\ \forall p\in \mathcal{P}_i\ \&\ p\in \textup{dom}(\mu_i)\,,\\
    & \int_p ds \lambda_i\leq \frac{1}{N},\ \forall p\in \mathcal{P}_j\ \&\ p\in \textup{dom}(\mu_i)\,,\\
    & \int_p ds \lambda_i\leq \frac{1}{N},\ \forall p\in \mathcal{P}_{N+1}\ \&\  p\in \textup{dom}(\mu_i)\,,
\end{split}
\end{equation}
for $j\neq i$. The optimal value for the total equals the sum of the volumes of the maximal slices anchored on each $\sigma_i(A)\cup \gamma_i$, plus the volume of the maximal slice anchored on the union $\sigma(A)=\displaystyle\cup_{i=1}^n \sigma_i(A)$; that is $\sum_{i=1}^N \max[\textup{Vol} (\Sigma_i(A))]+\max[\textup{Vol} (\Sigma(A))]$.

As in the bipartite case, $\Sigma(A)$ is not required to contain any of the individual HRT surfaces, so the volume of $\Sigma(A)$ is always greater than or equal to the sum of the volumes of the $\Sigma_i(A)$. However, even if all the individual HRT surfaces lie within $\Sigma(A)$, the presence of the `hole' guarantees that the inequality is strict whenever $N>2$. We refer to this property as \emph{soft superadditivity}:
\begin{equation}
    \max\left[\textup{Vol}\ \Sigma(A)\right] > \sum_{i=1}^N \max\left[\textup{Vol}\ \Sigma_i(A)\right]\,.
    \label{softsuperadditivity}
\end{equation}
This leads to a weaker inequality for the corresponding complexities, compared with the bipartite case,
\begin{equation}\label{npartitec}
\Delta \mathcal{C}^{(N)}\equiv \sum_{i=1}^N\mathcal{C}(\sigma_i(A))-\mathcal{C}(\sigma(A))< 0\,.
\end{equation}
We will call $\Delta \mathcal{C}^{(N)}$ the \emph{$N$-partite complexity}, thereby generalizing the notion of mutual complexity, given in (\ref{mutualcomplexity}).

To finish this section, we note that the above program can be naturally extended to compute the maximal volume of the surface anchored on the union of all HRT surfaces, that is, the surface defining the hole. Specifically, consider an arbitrary bulk spacelike surface $\tilde{\Sigma}_{N+1}$ anchored on $\displaystyle\cup_{i=1}^N \gamma_i$ and define its domain of dependence $D_{N+1}$. We introduce a new measure $\mu_{N+1}$ to compute quantities involving the latter bulk region, and denote the set of threads passing through it as $\mathcal{P}_{N+1}$. The program in this case would be given by

\begin{equation}
\begin{split}
    \min &\sum_{i=1}^{N+1}\left(\int_{\mathcal{P}_i}d\mu_i+\frac{1}{N+1}\int_{\mathcal{P}} d\mu_i\right)\\
    &\textup{s.t. }\rho_i\geq 1\ \forall i=1,...,N+1\,.
\end{split}
\label{lagrangianmultipleregions2}
\end{equation}
Following the same steps as before, we find in this situation a tighter version of superadditivity, which now allows for possible saturation:
\begin{equation}
\begin{split}
    &\qquad\qquad\max\left[\textup{Vol}(\Sigma(A))\right]\geq\\
    &\sum_{i=1}^N \max[\textup{Vol}( \Sigma_i(A))]
    +\max\left[\textup{Vol}(\Sigma_{N+1})\right].
\end{split}
\end{equation}
This inequality will generally not be saturated unless all the HRT surfaces lie within $\Sigma(A)$. Even so, it will be closer to saturation than its soft version, Eq.~\eqref{softsuperadditivity}, because the volume of the hole,  $\textup{Vol}( \Sigma_{N+1})$, always contributes positively to the right-hand side. We refer to this new inequality as \emph{tight superadditivity}.

It is interesting to ask what $\textup{Vol}(\Sigma_{N+1})$ would mean in terms of complexity. Very recently, Balasubramanian \emph{et al.} introduced the notion of binding complexity \cite{Balasubramanian:2018hsu}: the minimum number of inter-party gates required to generate a multipartite entangled state. Geometrically, those inter-party gates must run through the entanglement shadow, so a natural identification would be:
\begin{equation}
\mathcal{C}_{\text{binding}}=\max\left[\textup{Vol}(\Sigma_{N+1})\right]\,.
\end{equation}
That is, up to an overall normalization fixed by the CV dictionary (recall that we are working in units where $G_N\ell=1$), the hole's maximal volume counts precisely the non-local gates that `bind' the separate subregions into a single code subspace. Assuming this interpretation is correct, tight superadditivity then implies
\begin{equation}\label{tightSA}
\widetilde{\Delta \mathcal{C}}^{(N)}\equiv \Delta \mathcal{C}^{(N)}+\mathcal{C}_{\text{binding}} \leq 0\,,
\end{equation}
where $\Delta \mathcal{C}^{(N)}$ is given in (\ref{npartitec}), and $\widetilde{\Delta \mathcal{C}}^{(N)}$ is a generalized $N$-partite complexity that includes the binding term. This inequality is valid for $N \geq 3$. Together with (\ref{mutualcomplexity}), 
which applies to the bipartite case, these relations form part of what we dub the \emph{complexity cone}, i.e., a set of inequalities characterizing holographic CV complexity.

\section{Interpretation of the multi-flavor program\label{interpretation}}

We have succeeded in constructing a program that computes both the holographic complexity of the full state and the subsystem complexities of arbitrary subregions. However, the microscopic interpretation of the multiple flavors remains unclear. Here, we argue that this program naturally implements the nonlocal bulk operations required to satisfy superadditivity and, possibly, all other inequalities of the complexity cone.

\subsection{Discretization \& elementary gates}
In the previous sections, we have formulated a continuous program of Lorentzian threads with multiple flavors, satisfying certain density bounds. It is now instructive to think about the discretized version of such a program. For concreteness, we will focus our discussion below on the bipartite case, but our analysis can be easily generalized to include an arbitrary number of partitions.

After solving the optimization problem, one may characterize the solution by four 
numbers $N_i^j$, representing threads of type $i \in \{1,2\}$ (that is, those on which 
the measure $\mu_i$ acts non-trivially) crossing the domain of dependence $D_j$, 
with $j \in \{1,2\}$. One then finds that the subregion complexity of $\sigma_k$ is 
given by $\mathcal{C}(\sigma_k) = N^k_k$ (no sum over $k$), while the complexity of 
the full state is $\mathcal{C}(\sigma) = \frac{1}{2}\sum_{i,j} N_i^j$. The left panel of Fig.~\ref{fig:hyper} depicts the four basic threads that can appear in this program. The two subregions $A$ and $B$ are shown in blue and red, respectively, while threads of type $A$ and $B$ also follow the same color code. Threads of one kind that end in the `wrong' region (red threads that end in region $A$, or blue threads that end in region $B$) do not contribute to subregion complexities; only to the complexity of the full state. In the same figure, we have also written down the relative weights, $1$, $0$, or $1/2$, representing the individual contributions of each type of thread to the subregion complexities and to the complexity of the full system.

These four types of threads are sufficient to compute both the subsystem complexities and the complexity of the full state for a holographic system under arbitrary bipartitions. Thus, in analogy with \cite{Pedraza:2021mkh,Pedraza:2021fgp}, one can think of them as representing the set of elementary gates that enter into the definition of circuit complexity for a discretized bulk 
state. However, there are some aspects of this interpretation that are not completely satisfactory.

First, while these threads naively represent quasi-local operations that couple degrees of freedom confined to infinitesimal regions, the crossed threads do not seem to admit such a simple microscopic interpretation. In particular, they do not count towards any of the subregion complexities, so it is 
not clear whether they must attach to physical degrees of freedom in a tensor network discretizing the bulk state. More importantly, there is a fundamental issue with superadditivity. To see this, note that, for a given solution of the optimization problem, the four numbers must be nonnegative, $N_i^j \geq 0$. Other than that, they can generally be quite arbitrary, subject to the constraint imposed by superadditivity, i.e.,
\begin{equation}
\mathcal{C}(\sigma_1)+\mathcal{C}(\sigma_2)-\mathcal{C}(\sigma)
  = \frac{1}{2}\bigl(N_1^1+N_2^2-N_1^2-N_2^1\bigr) \leq 0\,.
\end{equation}
Whenever superadditivity is saturated, the number of `crossed' threads ($N_1^2+N_2^1$) must be equal to the number of non-crossed ones ($N_1^1+N_2^2$); more generally, the crossed ones must dominate in the solution. The problem here arises from micro-superadditivity. Namely, some of the individual gates (the non-crossed ones) contribute in the `wrong' way to the inequality, violating it in a 
microscopic sense.
Macroscopically, this violation is masked by the fact that crossed threads dominate in a given solution, so that the total contribution satisfies superadditivity, thanks to delicate cancellations. From the perspective of a microscopic gate set, this is conceptually unsatisfactory: one would like the elementary operations themselves to respect the basic complexity inequalities, rather than relying on combinations of gates that would individually drive the system 
outside the complexity cone.

Indeed, in the following subsections, we will show that the multi-flavor program can be recast in a way that satisfies superadditivity at the microscopic level. The price to pay, however, is that such a reformulation must explicitly include non-local gates that enable intrinsically non-local computation from the bulk perspective.

\subsection{Hyperthreads as a combination of gates?}

As a first conceptual exercise, we note that the program \eqref{lagrangianmultipleregions2} can be naturally extended by introducing a new type of object: a Lorentzian analogue of hyperthreads, originally defined in the Riemannian setting in \cite{Harper:2021uuq} (see Appendix~\ref{App: Riemannian threads and hyperthreads} for a brief review). Riemannian 
hyperthreads were devised to capture fundamental multipartite entanglement links, which are essential for holographic states in which the mostly-bipartite conjecture fails~\cite{Akers:2019gcv}. In the same spirit, we now postulate the existence of \emph{Lorentzian hyperthreads}, encoding additional elementary non-local gates that act simultaneously on several spatially separated degrees of freedom.

\begin{figure}
    \includegraphics[width=0.45\textwidth]{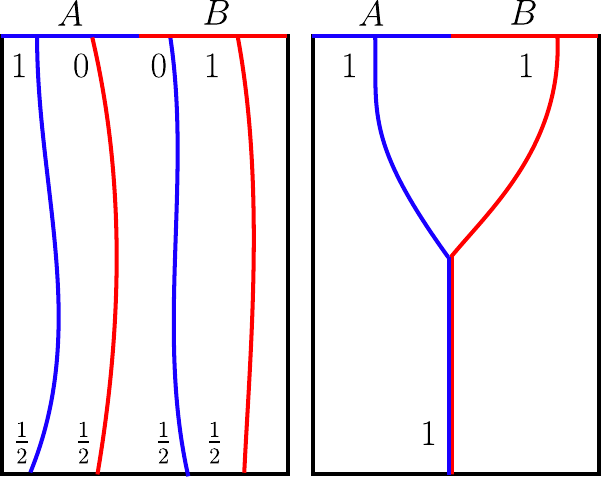}
    \caption{Left: elementary threads of the multi-flavor program. Right: A Lorentzian hyperthread obtained by combining two of these elementary threads. The resulting weights, +1, +1 and +1, signaling the contributions to subregion complexities and the complexity of the full state, are analogous to those in Riemannian hyperthreads. However, in the Lorentzian case, they lead to violations of micro-superadditivity.}
    \label{fig:hyper}
\end{figure}
 
A Lorentzian $k$-thread, or simply a $k$-thread, is an everywhere timelike curve in $\mathcal{M}$, that has $k$ legs, or branches.  We may visualize one of such objects as a thread anchored somewhere on $\partial \mathcal{M}^-$ which splits $k-1$ times, with each branch crossing a distinct domain of dependence $\{D_{1}(A),\ldots,D_{N}(A),D_{N+1}$\} before terminating somewhere on $\partial \mathcal{M}^+$. For an $N$-partite system, the maximal number of branches is then  $k_{\text{max}}=N+2$. We denote by $H_k$ the set of all $k$-threads with fixed $k$, and by $H$ the union of all $H_k$. For convenience, we also define $H_{D_i\ldots D j}$ as sets of $k$-threads which have branches passing through specific domains $D_i,...,D_j$. These can contain threads with values of $k$ strictly larger than the number of indices in their label. To make these definitions concrete, we now present a simple example:

\emph{Example:} Suppose that the boundary Cauchy slice $\sigma(A)$ is divided in three regions $\sigma_1(A),\ \sigma_2(A)$ and $\sigma_3(A)$. Since $N>2$ we have 4 domains of dependence $D_i$, with $i=1,\ldots4$, the last one corresponding to the central hole. In $H_{D_1}$ one finds those 2-threads going through $D_1(A)$ and those 3-, 4- and 5-threads with one branch crossing $D_1(A)$. In $H_{D_1,D_2}$ we find only those 3-, 4- and 5-threads with branches going through $D_1(A)$ and $D_2(A)$. The rest of the sets can be constructed in a similar manner.

Let us focus on the bipartite case, in which only 2- and 3-threads can appear. In this setting, there is only one RT surface, $\gamma_1=\gamma_2=\gamma$, so the central hole is absent. To account for $3$-threads, we introduce a measure $\mu_{1,2}$ that acts solely on $H_3$. The only way to incorporate this measure without changing the objective is as follows:
\begin{equation}
\begin{split}
    \!\!\min \bigg[&\int_{H_{D_1}}\!\!(d\mu_1+d\mu_{1,2})+\int_{H_{D_2}}\!\!(d\mu_2+d\mu_{1,2})\\
    &\quad+\frac{1}{2}\int_{H}(d\mu_1+d\mu_2+2d\mu_{1,2})\bigg]\\
    \!\!\!\textup{s.t.}\quad&\rho_1+\rho_{1,2}\geq 1\,,\ \rho_2+\rho_{1,2}\geq 1\ \,\, \forall x\in \mathcal{M}\,,
\end{split}
\label{programbipartite2}
\end{equation}
where $\rho_1$ and $\rho_2$ are defined as in \eqref{programbipartite}, and $\rho_{1,2}$ is given by an analogous expression. The factor of 2 in front of the measure $\mu_{1,2}$ is not arbitrary: complementary slackness~\cite{boyd2004convex} fixes this coefficient, as any other choice would prevent the 3-threads from contributing to the objective. The constraints, on the other hand, can be written as 
\begin{equation}
\rho_1+\rho_{1,2}\geq 1\,,\qquad\rho_2+\rho_{1,2}\geq 1\,,
\end{equation}
rather than imposing $\rho_1\geq 1$, $\rho_2\geq 1$, and $\rho_{1,2}\geq 1$ separately. Doing so would introduce an additional Lagrange multiplier, and ultimately modify the objective function.

With the formulation of the program justified, we now proceed to derive its dual. We begin with the Lagrangian:
\begin{equation}
\begin{split}
    L=&\int_{H_{D_1}}(d\mu_{1}+d\mu_{1,2})+\int_{H_{D_2}}(d\mu_2+d\mu_{1,2})\\
    &+\frac{1}{2}\int_{H}(d\mu_1+d\mu_2+2d\mu_{1,2})\\
    &+\int_M\lambda_1\left(1-\int_H (d\mu_1+d\mu_{1,2})\right)\\
    &+\int_M\lambda_2\left(1-\int_H (d\mu_2+d\mu_{1,2})\right)\,,\\
    =&\int_{\mathcal{M}}\lambda_1+\int_{H_{D_1}}d\mu_1\left(\frac{3}{2}-\int_h ds \lambda_1\right)\\
    &+\int_{H_{D_2}}d\mu_1\left(\frac{1}{2}-\int_h ds \lambda_1\right)+(1\leftrightarrow 2)\\
    &+\int_H d\mu_{1,2}\left(3-\int_h ds (\lambda_1+\lambda_2)\right),
\end{split}
\end{equation}
where the integral over $h$ denotes integration along an entire $k$-thread. The dual program is:
\begin{equation}\label{eq: dualbipartitehyperthread}
\begin{split}
      &\quad \max \int_\mathcal{M} d^dx\sqrt{-g}(\lambda_1+\lambda_2)\quad \textup{ s.t. }\\
&\int_{h}ds \lambda_{i}\leq 3/2\ \forall\ h\in \mathcal{P}_{i}\  \&\ h\in \textup{dom}(\mu_{i})\,, \\
  &  \int_{h}ds \lambda_{i}\leq 1/2\ \forall\ h\in \mathcal{P}_{j} \ \&\ h\in \textup{dom}(\mu_{i})\,,\\
 & \int_h ds (\lambda_1+\lambda_2)\leq 3\ \forall h\in H_3\ \&\ h\in \textup{dom}(\mu_{1,2})\,.
\end{split}
\end{equation}
In the second and third lines, $i,j = 1,2$, with the implicit understanding that $i \neq j$. A detailed examination of this dual program shows that, if a solution includes hyperthreads, it fails to reproduce the expected subregion complexities. To see this, note that any hyperthread that does not split in the region enclosed by $\Sigma_1(A) \cup \Sigma_2(A) \cup \Sigma(A)$ will cross a barrier of value four, so that the last constraint in \eqref{eq: dualbipartitehyperthread} is violated (see Fig.~\ref{fig:hyperthread}). Therefore, in the bipartite case, an optimal solution, if it exists, must either (i) contain only $2$-threads or (ii) contain $3$-threads, however, it would not yield a result equal to the sum of the subregion complexities, as in eq. \eqref{eq: bipartitesolution}.

\begin{figure}
    \centering
    \includegraphics[width=0.7\linewidth]{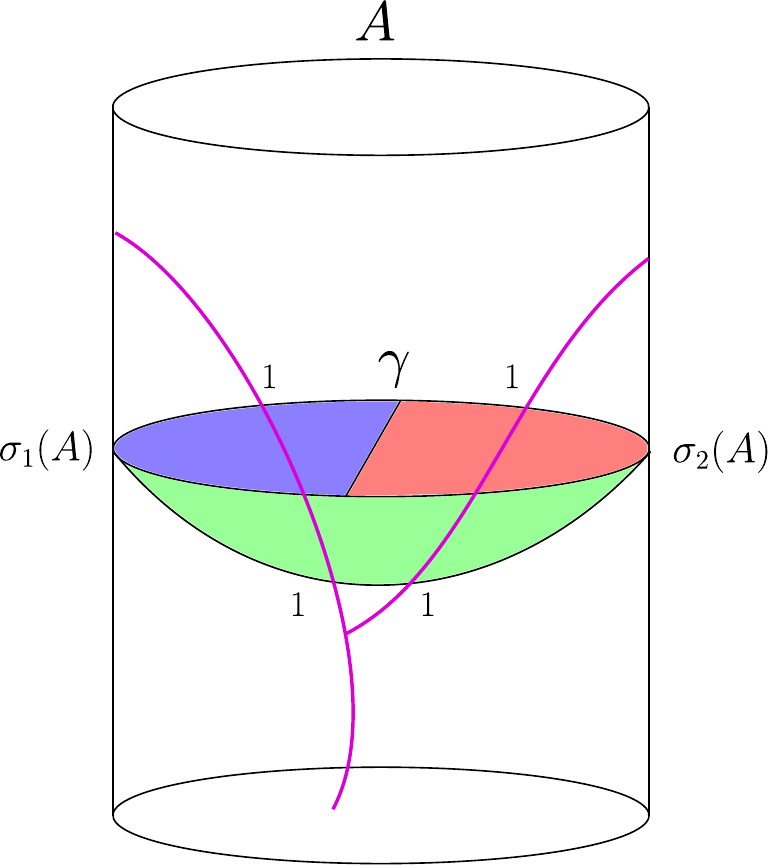}
    \caption{Visual representation of a hyperthread that violates the last constraint of \eqref{eq: dualbipartitehyperthread}. In this example, the hyperthread crosses level sets of total value of four.}
    \label{fig:hyperthread}
\end{figure}

This analysis can be generalized to an $N$-partite system, but the results are not 
particularly illuminating, since the solution suffers from the same issues as 
in the bipartite case. Hence, we will not repeat this analysis here. Nevertheless, 
we can offer some intuition as to why the inclusion of hyperthreads seems to yield 
inconsistent results in the Lorentzian case. The first observation is that, at the 
microscopic level, each hyperthread can be thought of as a combination of gates in 
the multi-flavor program. This is illustrated in Fig.~\ref{fig:hyper}, right panel. 
The composite operation represented by a hyperthread has weights $+1$, $+1$, and $+1$ 
that contribute to the complexities $\mathcal{C}(\sigma_1)$, $\mathcal{C}(\sigma_2)$, 
and $\mathcal{C}(\sigma)$, respectively, analogous to the Riemannian case. However, 
a close inspection of such operations reveals that they are fundamentally incompatible with 
basic complexity inequalities that follow from the CV proposal. In the Riemannian case, 
subadditivity of entanglement entropies requires 
\begin{equation}
S(A_1)+S(A_2)-S(A_1\cup A_2)\geq0\,,
\end{equation}
and each hyperthread contributes positively to the left-hand side of the inequality. 
However, for CV complexity, the relevant inequality is superadditivity, e.q.
(\ref{mutualcomplexity}),
\begin{equation}
\mathcal{C}(\sigma_1)+\mathcal{C}(\sigma_2)-\mathcal{C}(\sigma_1\cup \sigma_2)\leq0\,.
\end{equation}
As each Lorentzian hyperthread also contributes positively to the left-hand side, we conclude that they inevitably violate superadditivity at the microscopic level. 

The above observation does not fully rule out Lorentzian hyperthreads, as the non-crossed 2-threads also suffer from the same issue described above. One can imagine, for instance, that there are other elements in the elementary gate set that may contribute negatively to superadditivity and compensate for the existence of hyperthreads. If this is not the case, then the number of hyperthreads in an optimal solution would be negative, which would be in contradiction with the density bound. More importantly, hyperthreads are also unnatural from the perspective of an elementary gate basis, since they represent operations that violate micro-superadditivity. In the next subsection, we will argue that the elementary gates of the multi-flavor program can still be combined, in a suitable way, to resemble the structure of hyperthreads, while forming a good set of operations that respect superadditivity (\ref{mutualcomplexity}), or its $N$-partite generalization (\ref{tightSA}).

\subsection{Change of basis \& generalized hyperthreads}

The question we would like to address now is whether it is possible to construct a `good' gate set by combining the elementary gates of the multi-flavor program into a new physical basis. Such a basis should satisfy a few conditions: (i) it must be invertible, (ii) it must respect superadditivity at the microscopic level, and (iii) it should yield four non-negative numbers $\tilde N_i \ge 0$ for $\tilde N = \{\tilde N_1,\tilde N_2,\tilde N_3,\tilde N_4\}$, as a solution of the optimization problem. Some of these conditions go hand in hand. For example, hyperthreads cannot be elements of a physical basis because they violate superadditivity: if we force them to be part of the basis, an optimal solution would necessarily involve a negative number of hyperthreads. Aside from these requirements, we still have some freedom in choosing the new basis, as we will see below.

As a further constraint, we would like `standard' 2-threads to be part of the new basis. These can be obtained by combining one thread of each type that crosses the same domain of dependence. Such combinations generate the more common operations with the weights $+1$ and $0$ for the subregion complexities $\mathcal{C}(\sigma_i)$ and $+1$ for the complexity of the full state $\mathcal{C}(\sigma)$. See Fig.~\ref{fig:hyper2}, top panels. Note that these threads saturate superadditivity: individually, their contribution to the mutual complexity vanishes. In fact, this is the reason why the original program of Lorentzian threads \cite{Pedraza:2021mkh,Pedraza:2021fgp} cannot yield a negative mutual complexity, as explained around eq.~(\ref{mutualcomplexity}). As for the remaining two elements, we would like to have a combination that resembles a hyperthread (see Fig.~\ref{fig:hyper}, left panel), but has more weight for the complexity of the full state $\mathcal{C}(\sigma)$. A weight of $+2$ would saturate the inequality, as do the standard 2-threads. The minimal weight that yields a negative mutual complexity is therefore $+3$. In Fig.~\ref{fig:hyper2}, bottom panels, we show examples of such combinations, obtained by combining one of each non-crossed thread type and four crossed ones. We call these objects \emph{generalized hyperthreads}. Mathematically, we can express this change of basis by an algebraic equation of the form  $N= \mathbf{M}\cdot\,\tilde N$, for some appropriate transformation matrix $\mathbf{M}$, where we have defined $N\equiv\{N_1,N_2,N_3,N_4\}=\{N_1^1,N_2^1,N_1^2,N_2^2\}$. For example, for the basis depicted in Fig.~\ref{fig:hyper2}, we have:
\begin{equation}\label{Mexample1}
\begin{pmatrix}
N_1 \\ N_2 \\ N_3 \\ N_4 \\
\end{pmatrix}
=\begin{pmatrix}
1 & 0 & 1 & 1\\
1 & 0 & 4 & 0\\
0 & 1 & 0 & 4\\
0 & 1 & 1 & 1\\
\end{pmatrix}
\begin{pmatrix}
\tilde N_1 \\ \tilde N_2 \\ \tilde N_3 \\ \tilde N_4 \\
\end{pmatrix}\,.
\end{equation}
One can check that $\det \mathbf{M}\neq 0$, so the transformation is invertible.\footnote{Since the matrix is invertible, we can use $\mathbf{M}^{-1}$ to read off the number of generalized hyperthreads corresponding to an optimal configuration expressed in the original 2-thread basis.} Moreover, all elements in this new basis respect superadditivity at the microscopic level, and one can convince oneself that, for any triad $\mathcal{C}(\sigma_1)$, $\mathcal{C}(\sigma_2)$ and $\mathcal{C}(\sigma)$ that respects superadditivity, one could find solutions of the program such that $N_i \ge 0\,$ $\forall i$, where
\begin{align}
\mathcal{C}(\sigma_1)&=\tilde N_1+\tilde N_3+\tilde N_4\,,\nonumber\\
\mathcal{C}(\sigma_2)&=\tilde N_2+\tilde N_3+\tilde N_4\,,\\
\mathcal{C}(\sigma)&=\tilde N_1+\tilde N_2+3(\tilde N_3+\tilde N_4)\,.\nonumber
\end{align}

\begin{figure}
    \includegraphics[width=0.45\textwidth]{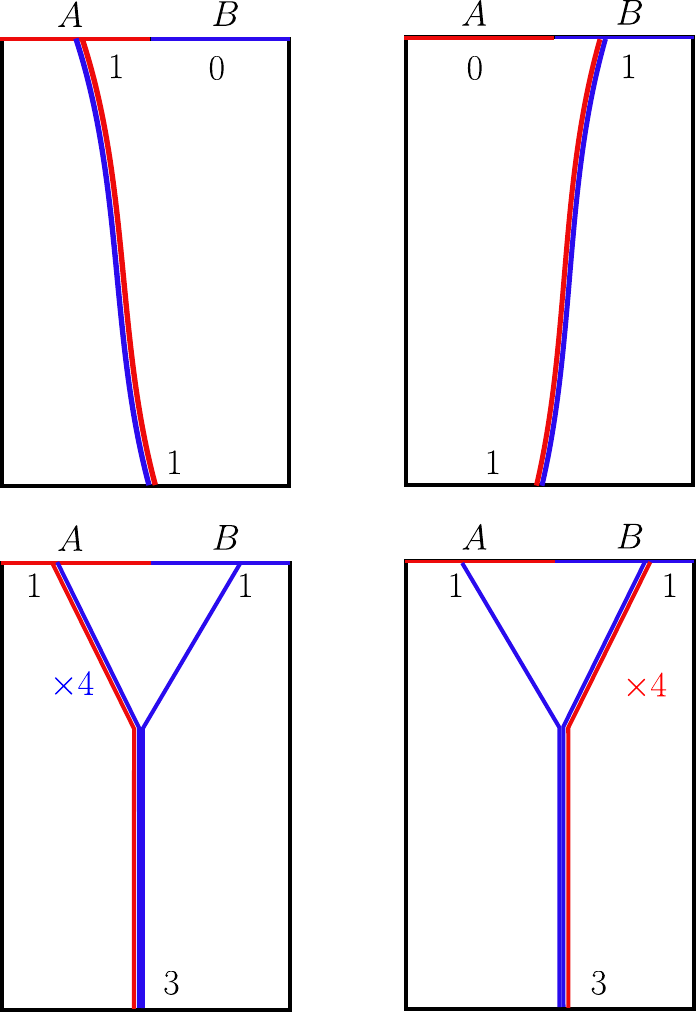}
    \caption{A possible change to a `good' basis. The new elementary operations include ordinary Lorentzian 2-threads (top two) and `generalized' hyperthreads (bottom two), with weights adjusted to respect superaditivity.}
    \label{fig:hyper2}
\end{figure}

A few comments are in order. First, we emphasize that the choice of physical basis is non-unique. The one above is reasonable from a physical standpoint, but there are also many other good choices. For example, the matrix
\begin{equation}
    \mathbf{M}=\begin{pmatrix}
1 & 0 & 1 & 1\\
1 & 0 & 3& 1\\
0 & 1 & 1 & 3\\
0 & 1 & 1 & 1\\
\end{pmatrix}
\end{equation}
also yields generalized hyperthreads with the same weights. More generally, we can increase the weight for $\mathcal{C}(\sigma)$ by an amount of order $\mathcal{O}(N^0)$ and still have a good basis, for example, one might consider
\begin{equation}
 \mathbf{M}=\begin{pmatrix}
1 & 0 & 1 & 1\\
1 & 0 & 6 & 0\\
0 & 1 & 0 & 6\\
0 & 1 & 1 & 1\\
\end{pmatrix}.
\end{equation}
In this last case one can only increase the mutual complexity by multiples of 2; however, since our system has a large-$N$ number of degrees of freedom, this will not affect the result at leading order in $1/N$. Thus, all of these are good transformations leading to a good basis, and are formally equivalent to the multi-flavor program

Second, it is important to emphasize the differences between generalized hyperthreads and normal hyperthreads. In particular, not only are the relative weights of the legs different, but so is the nature of the bulk interactions. For ease of visualization, we have depicted the generalized hyperthreads in Fig.~\ref{fig:hyper2} as normal hyperthreads. While it is true that they have $k$ external legs attached to the AdS boundary, the corresponding bulk configurations can be more general and flexible, and need not split locally at a single point in spacetime. Instead, one should think of them as extended, possibly nonlocal, networks of threads with $k$ boundary-anchored legs.

Finally, while we have only analyzed the bipartite case in detail, it is straightforward to see that the construction of a good basis for the $N$-partite case is always feasible. The key point is that the number of elementary gates in a basis increases with $N$, giving us more freedom to find suitable combinations. In the bipartite case we have four elements in a given basis, whereas for $N \ge 3$ we have $2(N+1)$ basic elements, due to the existence of the hole. A simple construction is to include $N+1$ standard 2-threads, each going through a single domain of dependence, and $N+1$ $(N+1)$-threads, with weights $+1,\ldots,+1$ for $\mathcal{C}(\sigma_i)$ ($\forall i$) and $\mathcal{C}_{\text{binding}}$, and $+(N+2)$ for $\mathcal{C}(\sigma)$. For example, in the tripartite case one such transformation that leads to a good basis is given by
\begin{equation}
\mathbf{M}=\begin{pmatrix}
1 & 0 & 0 & 0 & 1 & 1 & 1 & 1\\
1 & 0 & 0 & 0 & 6 & 0 & 0 & 0\\
0 & 1 & 0 & 0 & 1 & 1 & 1 & 1\\
0 & 1 & 0 & 0 & 0 & 6 & 0 & 0\\
0 & 0 & 1 & 0 & 1 & 1 & 1 & 1\\
0 & 0 & 1 & 0 & 0 & 0 & 6 & 0\\
0 & 0 & 0 & 1 & 1 & 1 & 1 & 1\\
0 & 0 & 0 & 1 & 0 & 0 & 0 & 6\\
\end{pmatrix},
\end{equation}
which directly generalizes (\ref{Mexample1}). This construction is always possible for arbitrary $N$, as can be shown by simple induction. More generally, for $N \ge 3$ one could consider more general bases, including hyperthreads with different $k$. We leave such generalizations to the interested reader.

\section{Discussion\label{discussion}}

In this work, we have revisited the Lorentzian thread formulation of holographic complexity and extended it to multipartite settings, with the goal of defining subregion complexity and exploring its structural constraints. Starting from the CV proposal, we showed that the original Lorentzian thread program \cite{Pedraza:2021mkh,Pedraza:2021fgp} cannot consistently reproduce the superadditivity of subregion complexity, even in simple bipartite configurations. The obstruction can be traced to the absence of a Lorentzian analogue of the multicommodity theorem: a single family of timelike flows is not flexible enough to compute, in one shot, the complexities of all subregions and of their union. To overcome this, we reformulated the problem in terms of measures and introduced a multi-flavor program, in which each boundary subregion is associated with its own family of Lorentzian threads. This enlarged framework yields a well-posed optimization problem that simultaneously computes the complexity of each subregion and of the full state. For general $N\geq3$ partitions, it leads to a set of inequalities, which we dub soft and tight superadditivity, the latter one including a binding contribution. All these inequalities represent (a part of) what we call the \emph{complexity cone} for CV.

A natural question is how this multi-flavor construction should be understood microscopically. To answer this, we analyzed the discretized version of the program and identified the corresponding set of elementary gates. For the bipartite case, the four basic thread types of the multi-flavor program can be viewed as quasi-local gatelines that couple degrees of freedom in a bulk tensor network representation of the subsystems and full CFT state. Taken at face value, these threads provide a complete gate set capable of generating both the subregion complexities and the complexity of the full state. However, we found that some of these elementary operations violate superadditivity at the microscopic level: individually, they contribute with the wrong sign to the mutual complexity, and only delicate cancellations among different gate types ensure that the macroscopic inequalities are satisfied. From the perspective of a microscopic description of complexity, this is conceptually unsatisfactory, as it would mean that the basic operations themselves drive the system outside the complexity cone.

We resolved this tension by exhibiting an explicit change of basis in the space of elementary operations. By taking appropriate linear combinations of the multi-flavor gatelines, one can construct a new physical basis in which each element respects micro-superadditivity while still reproducing the same macroscopic complexities. In this new basis, two of the elements are ordinary Lorentzian 2-threads, which saturate superadditivity and contribute only to the individual subregion complexities and to the full complexity in a balanced way. The remaining elements are \emph{generalized Lorentzian hyperthreads}: bulk objects with $k$ external legs whose weights are chosen so that they contribute positively to the complexity of the full state while inducing a negative mutual complexity. These generalized hyperthreads are intrinsically nonlocal operations, acting simultaneously on several spatially separated degrees of freedom in the dual CFT. Thus, the multi-flavor program can be recast as an optimization over a gate set that is manifestly compatible with the complexity cone, at the price of explicitly incorporating nonlocal gates. In this sense, our analysis refines the naive expectation that holographic states can be prepared using mostly local gates: once subregion complexities and their inequalities are imposed, nonlocal gates of generalized-hyperthread type appear to be an unavoidable part of the microscopic description, in line with the nonlocal computation protocols suggested by the connected wedge theorem \cite{May:2019odp,May:2022rko,May:2022clu,Dolev:2022gwj,May:2023kfp}.

Our results open up several directions for future work. A first, and conceptually close, avenue is to revisit the bit thread formulation of entanglement entropy in light of the multi-flavor structure uncovered here. In the Riemannian setting, the nesting property and the multicommodity theorem already underlie proofs of basic inequalities such as (strong) subadditivity and monogamy of mutual information, and they provide partial access to the holographic entropy cone. However, it is not clear to what extent the full set of holographic entropy inequalities can be realized by a single species of bit threads. The present work suggests introducing multiple kinds of threads---or a suitable generalization of hyperthreads---already at the level of entanglement entropy. It would be very interesting to understand whether a multi-flavor or hyperthread-enhanced bit-thread formalism can probe a larger portion of the entropy cone, or perhaps furnish a more direct geometric characterization of its facets.

A second natural direction is to extend our analysis of subregion complexity beyond the CV proposal. Recently, the Lorentzian thread framework has been generalized in \cite{Caceres:2023ziv} to a broad class of gravitational functionals collectively referred to as Complexity=Anything \cite{Belin:2021bga}. In that setting, different choices of bulk functional correspond to different candidate measures of complexity, all admitting a Lorentzian flow representation. Our multi-flavor construction can, in principle, be transplanted to these more general proposals, leading to a family of subregion complexity functionals and associated complexity cones. Comparing the resulting inequalities, and their degree of saturation, across different bulk prescriptions could provide a powerful diagnostic of how sensitive subregion complexity is to the detailed choice of bulk functional, and might help isolate robustness properties that any sensible holographic complexity measure should satisfy.

A third, more microscopic, line of investigation is to sharpen the connection between generalized Lorentzian hyperthreads and explicit models of quantum computation or quantum error correction \cite{Speelman:2016qua,Dolev:2022rzs,Allerstorfer:2023ycc,Asadi:2024fda,Bluhm:2025qgd}. The multi-flavor program suggests a concrete gate set---consisting of standard local gatelines and genuinely nonlocal hyperthread-like operations---that prepares holographic states while obeying the inequalities of the complexity cone. It would be valuable to construct tensor-network or circuit models in which these gates can be realized explicitly, and to compare the resulting circuit complexity with our geometric subregion complexities. Such models could clarify the role of binding-type contributions (associated to the hole), and might lead to a direct CFT definition of generalized hyperthreads as nonlocal logical gates acting on a code subspace.

More broadly, we expect the notion of a complexity cone to play a role analogous to that of the entropy cone, as a set of universal inequalities constraining which patterns of subregion complexity are compatible with a semiclassical bulk dual. In this paper we have taken a first step in this direction for CV, showing how a multi-flavor Lorentzian thread program and its generalized hyperthread interpretation naturally encode these constraints. It would be fascinating to see whether a more complete understanding of complexity cones, across different holographic prescriptions, can help map bulk observables to boundary observables and eventually provide a unifying organizing principle for holographic complexity measures and their field-theory duals. We hope to return to some of these questions in the near future.

\vspace{2mm}
 
\noindent \emph{Acknowledgments.}
We are grateful to J. Gamb\'in Egea, M. Headrick, V. Patil, A. Russo, A. Svesko and Z. Weller-Davies for many enlightening discussions and correspondence on related subjects throughout the last years.  The work of EC was supported by the National Science Foundation
under grant No. PHY–2210562 and by UT-CNS through a Spark 2025-2029 grant. RC and JFP are supported by the `Atracci\'on de Talento' program grant 2020-T1/TIC-20495 and by the Spanish Research Agency through the grants CEX2020-001007-S, PID2021-123017NB-I00 and PID2024-156043NB-I00, funded by MCIN/AEI/10.13039/501100011033, and by ERDF `A way of making Europe.' RC also acknowledges the additional support of the Spanish grant FPU with reference FPU22/01262.

\appendix

\section{Riemannian threads \& hyperthreads} \label{App: Riemannian threads and hyperthreads}

In this appendix, we review the reformulation of Riemannian bit threads in terms of measure theory, commonly employed in the study of entanglement entropy. In doing so, we follow closely \cite{Harper:2021uuq,Headrick:2022nbe}. Let $\Sigma$ be a Riemannian manifold with boundary. As a first step, we consider a bipartition of the boundary into two regions $A$ and $B$. According to the RT prescription \cite{Ryu:2006bv,Ryu:2006ef}, the entanglement entropy of region $A$ is given by the minimal area among all codimension-1 surfaces $m(A)$ homologous to $A$, that is (in units where $4G_N=1$)
\begin{equation}
    S(A)=\min [\textup{Area}(m(A))]\,.
\end{equation}

To employ the measure approach for computing the entanglement entropy, we define $\mathcal{P}$ as the set of curves with one endpoint in $A$ and the other in $B$. As in the Lorentzian case, it is necessary to use the delta function $\Delta(x,p)$ and the density function $\rho(x)$, both of which were defined earlier in \eqref{deltafunction} and \eqref{density}, respectively. For the reader's convenience, we recall their expressions here:
\begin{equation}
\begin{split}
    \Delta(x,p)=&\int_p ds \delta(x-y(s))\,,\\
     \rho(x)=&\int_{\mathcal{P}} d\mu \Delta(x,p)\,.
\end{split}
\end{equation}
The optimization program that computes the entropy is
\begin{equation}
    S(A)=\max \int_\mathcal{P} d\mu \textup{ s.t. }\rho(x)\leq 1, \forall x\in \Sigma\,.
\end{equation}
We now introduce a Lagrange multiplier $\lambda(x)$ to impose the density constraint:
\begin{equation}
\begin{split}
    L(\mu,\lambda)=&\int_{\mathcal{P}}d\mu +\int_{\Sigma}\lambda(x)\left(1-\int_{\mathcal{P}}d\mu\Delta(x,p)\right)\\
    =&\int_{\Sigma} \lambda(x)+\int_{\mathcal{P}}d\mu\left(1-\int_{p}ds\lambda(x)\right).
\end{split}
\end{equation}
From the second line, we can derive a dual optimization program, namely
\begin{equation}
    \min \int_\Sigma d^d x \sqrt{g} \lambda(x) \textup{ s.t. }\int_p ds \lambda\geq 1 \ \forall,\ p\in \mathcal{P}\,.
\end{equation}
Before proceeding further, it is necessary to ensure that both the primal and dual programs yield the same value. To this end, we have to verify that Slater's condition is satisfied. In particular, one possible solution to the primal problem consists of setting $\mu(p)=0\ \forall p\in \mathcal{P}$. Although this is not the optimal solution, the density bound is strictly satisfied, showing that Slater's condition (and consequently strong duality) holds.

The optimal solution corresponds to the configuration in which all level sets are located on the minimal-area surface homologous to $A$ \cite{Harper:2021uuq}. Any other configuration will either increase the objective or violate the density bound. Therefore,
\begin{equation}
    S(A)=\max \int_{\mathcal{P}}d\mu=\min [\textup{Area}(m(A))]\,.
\end{equation}
Up to this point, we have focused solely on region $A$. However, since $A$ and $B$ are complementary, they have exactly the same entanglement entropy.

If instead of partitioning the boundary of $\Sigma$ into two regions we partition it into $N$ regions, it is convenient to introduce the following sets:
\begin{equation}
\begin{split}
    \mathcal{P}_{ij}&=\{\textup{Threads joining }A_i \textup{ and } A_j\}\,,\\
    \mathcal{P}&=\{\textup{All threads}\}\,.
\end{split}
\end{equation}
The program that simultaneously computes all the entropies in terms of threads is
\begin{equation}
\begin{split}
    &\displaystyle \sum_{i=1}^N S(A_i)=\max \sum_{i,j=1}^N \int_{\mathcal{P}_{ij}}d\mu\,,\\
    &\textup{ s.t. } \rho(x)\leq 1,\ \forall x\in \Sigma\,.
\end{split}
\end{equation}

We can now generalize the notion of a thread in order to study multipartite entanglement. Consider a partition of the boundary into $N$ regions. In addition to the usual threads that join two different regions, we introduce a $k$-thread (or hyperthread) as a curve that splits at several points and has a total of $k$ endpoints (with $k\leq N$). It is important to emphasize that each endpoint must lie in a different region and cannot end in the bulk. The set of all such threads, denoted by $H$, can be decomposed as the union of $k$-threads joining the $k$ regions $A_i, A_j, \dots, A_k$, which we denote by $H_{A_i A_j\dots A_k}$. Similarly, the set of $k$-threads with one endpoint in $A_i$ will be written as $H_{i}$. Finally, $H_k$ denotes the set of all $k$-threads.

The program that allows for the computation of the entropies of all regions is then given by
\begin{equation}
    \begin{split}
    \displaystyle \sum_{i=1}^N S(A_i)=\max\sum_{i=1}^N \int_{H_{\sigma_i}}d\mu \textup{ s.t. }\rho(x)\leq1\ \forall x\in \Sigma\,, \\
=\max\sum_{i=2}^N k\int_{H_k}d\mu \textup{ s.t. }\rho(x)\leq1\ \forall x\in \Sigma\,.
    \end{split}
\end{equation}
The Lagrangian corresponding to this program is
\begin{equation}
    \begin{split}
    L=&\displaystyle \sum_{i=2}^N k\int_{H_k}d\mu+\int_{\Sigma} \lambda(1-\rho(x))\\
    =&\int_{\Sigma} \lambda +\sum_{k=2}\int_{H_k}\left(k-\int_h ds \lambda\right).
    \end{split}
\end{equation}
The second line provides a new dual program:
\begin{equation}
    \min \int_{\Sigma}d^dx\sqrt{g}\lambda(x)\textup{ s.t. } \int_h ds\lambda(x)\geq k\ \forall h\in H_k\,.
\end{equation}
This is just a generalization of the program with 2-threads only, so the optimal solution is exactly the same and equal to the sum of the areas of the surfaces homologous to each region. One can show that, although the objective is not modified with respect to the original program, these hyperthreads have a maximal contribution which is non-trivial. In particular, this contribution has been computed in the case where only $2$- and $N$-threads are present. The contribution of each $N$-thread is
\begin{equation}
    \frac{1}{\alpha}\textup{area}(t_n)\,,
\end{equation}
where $t_n$ is the maximal area that can only be crossed by $N$-threads but not by $2$-threads and $\alpha$ is the number of times this surface is crossed. In this case, the contribution of each element in $H_2$ is reduced to
\begin{equation}
    \frac{1}{2}\left(\textup{area}(m_2)-\frac{n}{\alpha}\textup{area}(t_n)\right)\,.
\end{equation}

\section{Proofs of theorems}\label{app: proof}

The objective of this appendix is to provide a proof of Theorems 1 and 2, used in the main text. To this end, we follow an approach similar to that presented in \cite{Headrick:2022nbe}. For the sake of clarity, we restate the theorems here:

\thm*

\begin{proof}
Assume that condition \eqref{eq: condition1} holds. Consequently, we have:
\begin{equation}
\begin{split}
    1=&\psi|_A^B=\int_p d\psi =\int_p ds \frac{d\psi}{ds}=\int_p ds \frac{d y^\mu}{ds}\partial_\mu \psi\\
    =&\int_p ds \frac{d y^\mu}{ds}\partial^\nu \psi g_{\mu\nu}\,,
\end{split}
\end{equation}
where $y^\mu$ represents the trajectory of the thread. Before proceeding further, it is important to note that $d\psi$ is timelike and future directed. The vector field dual to this one ($\partial^\mu \psi\partial_\mu$) will also be timelike but past directed. Since the thread is future directed, the vector field $\frac{dy^\mu}{ds}\partial_\mu$ will likewise be future directed. It can be shown that, under these conditions, $g_{\mu\nu}u^\mu v^\nu\geq |u||v|$ for $u$ timelike, future directed, and $v$ timelike, past directed. Hence,
\begin{equation}
    1\geq \int_p ds\left|\frac{d y^\mu}{ds}\right||d\psi| \,.
\end{equation}
Given that, by the definition of proper distance, $\left|\frac{dy^\mu}{ds}\right|=1$, it follows that
\begin{equation}
    1\geq \int_p ds|d\psi| \geq \int_p ds\lambda\,.
\end{equation}
In other words, $\int_p ds \lambda \leq 1$, leading to \eqref{eq: condition2}.\\

We now show the converse implication. Suppose \eqref{eq: condition2} holds. We can express the integral over the thread as $\int_p ds \lambda=\int_p dt |-\dot{x}|\lambda$, where $-\dot{x}$ is the covector associated with the tangent to the curve $p$. The minus sign is included because the tangent vector is future directed. We define
\begin{equation}
\begin{split}
    &\psi_-(y):=\sup_{\substack{p\ \textup{timelike}\\\textup{from } A\ \textup{to}\ y}}\int_p dt |-\dot{x}|\lambda\,,\\
    &\psi_+(y):=\sup_{\substack{p\ \textup{timelike}\\\textup{from } y\ \textup{to}\ B}}\int_p dt |-\dot{x}|\lambda\,. 
\end{split}
\end{equation}
where the supremum is taken over all curves going from $A$ to $y$ and from $y$ to $B$, respectively. By assumption,
\begin{equation}
    \psi_-(y)+\psi_+\leq 1\,,
\end{equation}
and
\begin{equation}
    \lim_{y\rightarrow A} \psi_-(y)=0,\ \ \  \lim_{y\rightarrow B} \psi_+(y)=0\,.
\end{equation}
Let us now calculate the gradient of $\psi_-$. Provided that the integrand is a differentiable function of $\dot{x}$ for timelike, future-directed curves, the Hamilton–Jacobi formula\footnote{Hamilton–Jacobi formula asserts that the variation of the on-shell action with respect to changes in the final position is given by the canonical momentum at that point.} can be applied \cite{Headrick:2022nbe}. Although the optimal thread may be lightlike at some points, these issues can be resolved by extending the domain to the entire tangent space and constraining the integrand to be $-\infty$ whenever the velocity $-\dot{x}$ lies outside the future light cone.
\begin{displaymath}
\psi_-(y)=\sup_{\substack{q\ \textup{timelike}\\\textup{from } A\ \textup{to}\ y}}\int_q dt\left\{ \begin{array}{ll}
|-\dot{x}|\lambda & -\dot{x}\in \mathfrak{j}^+\\
-\infty &  \textup{otherwise} \\
\end{array}
\right.,
\end{displaymath}
\begin{displaymath}
\psi_+(y)=\sup_{\substack{q\ \textup{timelike}\\\textup{from } y\ \textup{to}\ B}}\int_q dt\left\{ \begin{array}{ll}
|-\dot{x}|\lambda & -\dot{x}\in \mathfrak{j}^+\\
-\infty &  \textup{otherwise} \\
\end{array}
\right. ,
\end{displaymath}
where $\mathfrak{j}^+$ denotes the set of timelike and future-directed covectors. When $-\dot{x}$ is timelike (which is the case of interest for us), we find that $\pi_{\pm\mu}=\partial_{\dot{x}^\mu}(|-\dot{x}|\lambda)=-\lambda \dot{x}_{\mu}/|-\dot{x}|$. Therefore,
\begin{equation}
    |d\psi_\pm|^2=\frac{\lambda^2 \dot{x}_\mu\dot{x}^\mu}{|-\dot{x}|^2}\geq \lambda^2\,.
\end{equation}

Thus $|d\psi_\pm|\geq \lambda$. We choose the following combination of these two functions to define $\psi(x)$:
\begin{equation}
    \psi(x)=\frac{\psi_-(x)-\psi_+(x)}{2(\psi_-(x)+\psi_+(x))}\,,
\end{equation}
whose exterior derivative is
\begin{equation}
    d\psi=\frac{\psi_+}{(\psi_-+\psi_+)^2}d\psi_-+\frac{\psi_-}{(\psi_-+\psi_+)^2}(-d\psi_+)\,.
\end{equation}
Since $\psi_-$ increases along a timelike curve, $d\psi_-$ will be future directed, while, as $\psi_+$ decreases along it, $d\psi_+$ must be past directed. Taking into account the reverse triangle inequality \cite{Minguzzi:2019mbe}, we find
\begin{equation}
    \begin{split}
    |d\psi|&\geq \frac{\psi_+}{(\psi_-+\psi_+)^2}|d\psi_-|+\frac{\psi_-}{(\psi_-+\psi_+)^2}|d\psi_+|\\
    &\geq \frac{1}{(\psi_-+\psi_+)}\lambda\geq \lambda\,.
    \end{split}
\end{equation}
In the second inequality we used that $|d\psi_\pm|\geq \lambda$, and in the last one that $\psi_-+\psi_+\leq 1$. Therefore, we have shown that \eqref{eq: condition2} implies \eqref{eq: condition1}.
\end{proof}

\segundo*

\begin{proof}
The proof of this theorem closely follows that of the previous one. We begin by assuming that condition \eqref{eq: condition11} holds. For any thread $p_B \in \PP_B$, one has
    \begin{equation}
        \begin{split}
        1/2&=\psi|_C^{\partial D_2^+}=\int_{p_2}d\psi=\int_{p_2} ds \frac{d\psi}{ds}=\int_{p_2} ds \frac{dy^\mu }{ds}\partial^\nu \psi g_{\mu\nu}\\
        &\geq \int_{p_2} ds \left|\frac{dy^\mu }{ds}\right| |d\psi|\geq \int_{p_2} ds \lambda\,.
        \end{split}
    \end{equation} 
A similar calculation shows that $\!\int_{p_1} ds\lambda\leq 3/2,\ \forall p_1\in \PP_1$.

To show the converse, we define the functions
\begin{align}
    \notag&\psi_-(y):=\sup_{\substack{p\ \textup{timelike}\\\textup{from } C\ \textup{to}\ y}}\int_p dt |-\dot{x}|\lambda\,,\\
    &\psi_{1}(y):=\sup_{\substack{p_1\in \PP_1\\\textup{from } y\ \textup{to}\ \partial D^+_1}}\int_p dt |-\dot{x}|\lambda\,,\\ 
    \notag&\psi_{2}(y):=\sup_{\substack{p_2\in \PP_2\\\textup{from } y\ \textup{to}\ \partial D^+_2}}\int_p dt |-\dot{x}|\lambda\,. 
\end{align}
The first function $\psi_-(y)$ is well defined for the whole region $\mathcal{M}^-$, in contrast to $\psi_1$ and $\psi_2$, whose domains are $J^-(D_1(A))$ and $J^+(D_2(A))$, respectively. From \eqref{eq: condition22}, it follows that
\begin{equation}
    \begin{split}
        \psi_-(y)+\psi_1(y)\leq 3/2,\quad \forall y\in J^-(D_1(A))\,,\\
        \psi_-(y)+\psi_2(y)\leq 1/2,\quad \forall y\in J^-(D_2(A))\,,\\
    \end{split}
\end{equation}
and $\lim_{y\rightarrow C}\psi_-(y)=0$, $\lim_{y\rightarrow D_{1,2}}\psi_{1,2}(y)=0$. Again, to apply the Hamilton–Jacobi formula, we enlarge the domain of these functions to allow the existence of lightlike threads.
\begin{displaymath}
\psi_-(y)=\sup_{\substack{q\ \textup{timelike}\\\textup{from } C\ \textup{to}\ y}}\int_q dt\left\{ \begin{array}{ll}
|-\dot{x}|\lambda & -\dot{x}\in \mathfrak{j}^+\\
-\infty &  \textup{otherwise} \\
\end{array}
\right.,
\end{displaymath}
\begin{displaymath}
\psi_{1,2}(y)=\sup_{\substack{p_{1,2}\in \PP_{1,2}\\\textup{from } y\ \textup{to}\ \partial D_{1,2}^+}}\int_q dt\left\{ \begin{array}{ll}
|-\dot{x}|\lambda & -\dot{x}\in \mathfrak{j}^+\\
-\infty &  \textup{otherwise} \\
\end{array}
\right. .
\end{displaymath}

It is immediate to check that $|d\psi_-|$ and $|d\psi_{1,2}|$ are greater than $\lambda$, so one can define the function $\psi$ as follows
\begin{widetext}
\begin{equation}
    \psi(y)=\left\{ \begin{array}{ll}
\frac{\psi_--\psi_2}{4(\psi_-+\psi_2)} & \forall y\in J^-(D_2(A))\backslash J^-(D_1(A))\\
\frac{5\psi_--\psi_1}{4(\psi_-+\psi_1)} & \forall y\in J^-(D_1(A))\backslash J^-(D_2(A))\\
\min\left(\frac{\psi_--\psi_2}{4(\psi_-+\psi_2)},\frac{5\psi_--\psi_1}{4(\psi_-+\psi_1)}\right) & \textup{otherwise}
\end{array}
\right. .
\end{equation}
\end{widetext}
This function is continuous and satisfies all the requirements in \eqref{eq: condition11}.
\end{proof}

\bibliographystyle{apsrev4-2}
\bibliography{LTsRefs}

\end{document}